\pdfoutput=1

\documentclass[prodmode,acmtalg]{acmsmall}

\usepackage{graphicx}
\usepackage{color}
\usepackage{amsmath}
\usepackage{amssymb}

\usepackage[noline]{algorithm2e}
\DontPrintSemicolon
\SetFuncSty{sc}
\SetAlFnt{\small}
\SetAlCapFnt{\small}
\SetAlCapNameFnt{\small}
\SetAlCapSkip{4pt}
\SetKw{lBegin}{begin}
\SetKw{lEnd}{end}

\usepackage{epsfig}
\usepackage{times}

\newcommand{\email}[1]{{\tt #1}}

\newcommand{\Thta}{\mathrm{\Theta}}
\newcommand{\Omga}{\mathrm{\Omega}}
\newcommand{\bigO}{\mathrm{O}}

\newcommand{\RPi}{\mathrm{\Pi}}
\newcommand{\es}{\{\}}
\newcommand{\eq}{[\,]}
\newcommand{\sm}{-}
\newcommand{\nlv}{\textit{null}}
\newcommand{\E}{\mathrm{E}}
\newcommand{\outl}{\mathit{out}}
\newcommand{\inl}{\mathit{in}}
\newcommand{\foutl}{\mathit{first\mbox{-}out}}
\newcommand{\finl}{\mathit{first\mbox{-}in}}
\newcommand{\noutl}{\mathit{next\mbox{-}out}}
\newcommand{\ninl}{\mathit{next\mbox{-}in}}
\newcommand{\position}{\mathit{position}}
\newcommand{\vertex}{\mathit{vertex}}
\newcommand{\inject}{\mathit{inject}}
\newcommand{\pop}{\mathit{pop}}
\newcommand{\find}{\mathit{find}}
\newcommand{\unite}{\mathit{unite}}

\title{Incremental Cycle Detection, Topological Ordering, and Strong Component Maintenance}

\author{BERNHARD HAEUPLER \affil{Massachusetts Institute of Technology}
TELIKEPALLI KAVITHA \affil{Tata Institute of
Fundamental Research}
ROGERS MATHEW \affil{Indian Institute of Science}
SIDDHARTHA SEN \affil{Princeton University} 
ROBERT E. TARJAN \affil{Princeton University \& HP Laboratories}}

\begin{abstract}

We present two on-line algorithms for maintaining a topological order
of a directed $n$-vertex acyclic graph as arcs are added, and
detecting a cycle when one is created.  Our first algorithm handles
$m$ arc additions in $\bigO(m^{3/2})$ time.  For sparse graphs ($m/n =
\bigO(1)$), this bound improves the best previous bound by a logarithmic
factor, and is tight to within a constant factor among algorithms
satisfying a natural {\em locality} property.  Our second algorithm
handles an arbitrary sequence of arc additions in $\bigO(n^{5/2})$ time.
For sufficiently dense graphs, this bound improves the best previous
bound by a polynomial factor.  Our bound may be far from tight: we show that
the algorithm can take $\Omga(n^22^{\sqrt{2\lg n}})$ time by
relating its performance to a generalization of the $k$-levels problem of
combinatorial geometry. A completely different algorithm running in $\Thta(n^2\log n)$ time was given
recently by Bender, Fineman, and Gilbert. We extend both of our algorithms to
the maintenance of strong components, without affecting the asymptotic time
bounds.

\end{abstract}


\acmformat{Haeupler, B., Kavitha, T., Mathew, R., Sen, S., and Tarjan, R. E. 
2011. Incremental Cycle Detection, Topological Ordering, and Strong Component
Maintenance.}

\begin{document}

\setcounter{page}{1}

\begin{bottomstuff}
A preliminary version of this article appeared in {\it Proceedings of the
35$^{\text{th}}$ International Colloquium on Automata, Languages, and
Programming (ICALP), Reykjavik, Iceland, July 7-11, 2008.}\newline
Author's addresses:  Bernhard Haeupler, CSAIL, Massachusetts Institute of
Technology, Cambridge, MA 02139, United States, \email{haeupler@mit.edu}; work
done while the author was a visiting student at Princeton University.  Telikepalli
Kavitha, Tata Institute of Fundamental Research, Mumbai, India,
\email{kavitha@tcs.tifr.res.in}; work done while the author was at Indian
Institute of Science.  Rogers Mathew, Indian Institute of Science, Bangalore,
India, \email{rogers@csa.iisc.ernet.in}.  Siddhartha Sen, Department of Computer Science, Princeton University, Princeton, NJ 08540, United States, \email{sssix@cs.princeton.edu}.  Robert E. Tarjan, Department of Computer
Science, Princeton University, Princeton, NJ 08540, United States and HP
Laboratories, Palo Alto, CA 94304, United States, \email{ret@cs.princeton.edu}.\newline
Research at Princeton University partially supported by NSF grants CCF-0830676
and CCF-0832797. The information contained herein does not necessarily reflect
the opinion or policy of the federal government and no official endorsement
should be inferred.
\end{bottomstuff}
\maketitle

\section{Introduction} \label{sec:intro}

In this paper we consider three related problems on dynamic directed graphs: cycle
detection, maintaining a topological order, and maintaining strong components.  We
begin with a few standard definitions.  A {\em topological order} of a directed graph
is a total order ``$<$'' of the vertices such that for every arc $(v, w)$, $v < w$.
A directed graph is {\em strongly connected} if every vertex is reachable from every
other.  The {\em strongly connected components} of a directed graph are its maximal
strongly connected subgraphs.  These components partition the vertices
\cite{Harary1965}.  Given a directed graph $G$, its {\em graph of strong components}
is the graph whose vertices are the strong components of $G$ and whose arcs are
all pairs $(X, Y)$ with $X \ne Y$ such that there is an arc in the original graph from a vertex
in $X$ to a vertex in $Y$. The graph of strong components is acyclic
\cite{Harary1965}.

A directed graph has a topological order (and in general more than one) if and only
if it is acyclic.  The first implication is equivalent to the statement that every
partial order can be embedded in a total order, which, as Knuth~\cite{Knuth1973}
noted, was proved by Szpilrajn~\cite{Szpilrajn1930} in 1930, for infinite as well as
finite sets.  Szpilrajn remarked that this result was already known to at least
Banach, Kuratowski, and Tarski, though none of them published a proof.

Given a fixed $n$-vertex, $m$-arc graph, one can find either a cycle or a topological
order in $\bigO(n + m)$ time by either of two methods: repeated deletion of sources
(vertices of in-degree zero)~\cite{Knuth1973,Knuth1974} or depth-first search
\cite{Tarjan1972}.  The former method (but not the latter) extends to the enumeration
of all possible topological orders~\cite{Knuth1974}.  One can find strong components,
and a topological order of the strong components in the graph of strong components,
in $\bigO(n + m)$ time using depth-first search, either one-way
\cite{Cheriyan1996,Gabow2000,Tarjan1972} or two-way~\cite{Sharir1981,Aho1983}.

In some situations the graph is not fixed but changes over time.  An {\em
incremental} problem is one in which vertices and arcs can be added; a {\em
decremental} problem is one in which vertices and arcs can be deleted; a {\em
(fully) dynamic} problem is one in which vertices and arcs can be added or
deleted. Incremental cycle detection or topological ordering occurs in circuit evaluation
\cite{Alpern1990}, pointer analysis~\cite{Pearce2003}, management of compilation
dependencies~\cite{Marchetti1993,Omohundro1992}, and deadlock detection
\cite{Belik1990}.  In some applications cycles are not fatal; strong components, and
possibly a topological order of them, must be maintained.  An example is
speeding up pointer analysis by finding cyclic relationships~\cite{Pearce2003b}.

We focus on incremental problems.  We assume that the vertex set is fixed and
given initially, and that the arc set is initially empty.  We denote by $n$
the number of vertices and by $m$ the number of arcs added.  For simplicity in stating time
bounds we assume that $m = \Omga(n)$.  We do not allow multiple arcs, so $m \le
{n \choose 2}$.  One can easily extend our algorithms to support vertex additions in
$\bigO(1)$ time per vertex addition.  (A new vertex has no incident arcs.)  Our
topological ordering algorithms, as well as all others in the literature, can handle
arc deletions as well as insertions, since an arc deletion preserves topological
order, but our time bounds are no longer valid.  Maintaining strong components as
arcs are deleted, or inserted and deleted, is a harder problem, as is maintaining the
transitive closure of a directed graph under arc insertions and/or deletions.  These
problems are quite interesting and much is known, but they are beyond the scope of
this paper.  We refer the interested reader to Roditty and Zwick~\cite{RodittyZ2008}
and the references given there for a thorough discussion of results on these
problems.

Our goal is to develop algorithms for incremental cycle detection and topological
ordering that are significantly more efficient than running an algorithm for a static
graph from scratch after each arc addition.  In Section~\ref{sec:lim-search} we
discuss the use of graph search to solve these problems, work begun by Shmueli
\cite{Shmueli1983} and realized more fully by Marchetti-Spaccamela et
al.~\cite{Marchetti1996}, whose algorithm runs in $\bigO(nm)$ time.  In
Section~\ref{sec:2way-search} we develop a two-way search method that we call
{\em compatible search}.  Compatible search is essentially a generalization of
two-way ordered search, which was first proposed by Alpern et
al.~\cite{Alpern1990}.  They gave a time bound for their algorithm in an
incremental model of computation, but their analysis does not give a good bound
in terms of $n$ and $m$.  They also considered batched arc additions.  Katriel
and Bodlaender~\cite{Katriel2006} gave a variant of two-way ordered search with
a time bound of $\bigO(\min\{m^{3/2}\log n, m^{3/2} + n^2\log n\})$. Liu and
Chao~\cite{Liu2007} improved the bound of this variant to $\Thta(m^{3/2} +
mn^{1/2}\log n)$, and Kavitha and Mathew~\cite{Kavitha2007} gave another variant with a bound of
$\bigO(m^{3/2} + nm^{1/2}\log n)$.

A two-way search need not be ordered to solve the topological ordering problem. We
apply this insight in Section~\ref{sec:soft-search} to develop a version of
compatible search that we call {\em soft-threshold search}.  This method uses
either median-finding (which can be approximate) or random sampling in place of
the heaps (priority queues) needed in ordered search, resulting in a time bound
of $\bigO(m^{3/2})$.  We also show that any algorithm among a natural class of
algorithms takes $\Omga(nm^{1/2})$ time in the worst case.  Thus for sparse
graphs ($m/n = \bigO(1)$) our bound is best possible in this class of algorithms.

The algorithms discussed in Sections~\ref{sec:2way-search} and~\ref{sec:soft-search}
have two drawbacks.  First, they require a sophisticated data structure, namely a
{\em dynamic ordered list}~\cite{Bender2002,Dietz1987}, to maintain the topological
order.  One can address this drawback by maintaining the topological order as an
explicit numbering of the vertices from $1$ through $n$.  Following Katriel
\cite{Katriel2004}, we call an algorithm that does this a {\em topological sorting}
algorithm.  The one-way search algorithm of Marchetti-Spaccamela et
al.~\cite{Marchetti1996} is such an algorithm.  Pearce and Kelly~\cite{Pearce2006}
gave a two-way-search topological sorting algorithm.  They claimed it is fast in
practice, although they did not give a good time bound in terms of $n$ and $m$.
Katriel~\cite{Katriel2004} showed that any topological sorting algorithm that has a
natural {\em locality} property takes $\Omga(n^2)$ time in the worst case even if
$m/n = \Thta(1)$.

The second drawback of the algorithms discussed in
Sections~\ref{sec:2way-search} and~\ref{sec:soft-search} is that using graph
search to maintain a topological order becomes less and less efficient as the
graph becomes denser. Ajwani et al.~\cite{Ajwani2006} addressed this drawback
by giving a topological sorting algorithm with a running time of
$\bigO(n^{11/4})$. In Section~\ref{sec:top-search} we simplify and improve this
algorithm. Our algorithm searches the topological order instead of the graph. 
We show that it runs in $\bigO(n^{5/2})$ time.  This bound may be far from
tight.  We obtain a lower bound of $\Omga(n2^{\sqrt{2\lg n}})$ on
the running time of the algorithm by relating its efficiency to a generalization
of the $k$-levels problem of combinatorial geometry.

In Section~\ref{sec:strong} we extend the algorithms of Sections
\ref{sec:soft-search} and~\ref{sec:top-search} to the incremental maintenance of
strong components.  We conclude in Section~\ref{sec:remarks} with some remarks and
open problems.

This paper is an improvement and extension of a conference paper
\cite{Haeupler2008b}, which itself is a combination and condensation of two on-line
reports~\cite{Haeupler2008,Kavitha2007}.  Our main improvement is a simpler analysis
of the algorithm presented in Section~\ref{sec:top-search} and originally in
\cite{Kavitha2007}.  At about the same time as~\cite{Kavitha2007} appeared and also
building on the work of Ajwani et al., Liu and Chao~\cite{Liu2008} independently
obtained a topological sorting algorithm that runs in $\bigO(n^{5/2}\log^2 n)$ or
$\bigO(n^{5/2}\log n)$ time, depending on the details of the implementation. More
recently, Bender, Fineman, and Gilbert~\cite{Bender2009} have presented a topological
ordering algorithm that uses completely different techniques and runs in
$\Thta(n^2\log n)$ time.

\section{One-Way Search} \label{sec:lim-search}

\SetKwFunction{limitedsearch}{Limited-Search}

The simplest of the three problems we study is that of detecting a cycle when an arc
addition creates one.  All the known efficient algorithms for this problem, including
ours, rely on the maintenance of a topological order.  When an arc $(v, w)$ is added,
we can test for a cycle by doing a search forward from $w$ until either reaching $v$
(there is a cycle) or visiting all vertices reachable from $w$ without finding $v$.
This method takes $\Thta(m)$ time per arc addition in the worst case, for a total of
$\Thta(m^2)$ time.  By maintaining a topological order, we can improve this method.
When a new arc $(v, w)$ is added, test if $v < w$.  If so, the order is still
topological, and the graph is acyclic. If not, search for $v$ from $w$.  If the
search finishes without finding $v$, we need to restore topological order, since (at
least) $v$ and $w$ are out of order.  We can make the order topological by moving all
the vertices visited by the search to positions after all the other vertices, and
ordering the visited vertices among themselves topologically.

We need a way to represent the topological order. A simple numbering scheme suffices.
Initially, number the vertices arbitrarily from $1$ through $n$ and initialize a
global counter $c$ to $n$.  When a search occurs, renumber the vertices visited by
the search consecutively from $c + 1$, in a topological order with respect to the
subgraph induced by the set of visited vertices, and increment $c$ to be the new
maximum vertex number.  One way to order the visited vertices is to make the search
depth-first and order the vertices in reverse postorder \cite{Tarjan1972}. With this
scheme, all vertex numbers are positive integers no greater than $nm$.

Shmueli \cite{Shmueli1983} proposed this method as a heuristic for cycle detection,
although he used a more-complicated two-part numbering scheme and he did not mention
that the method maintains a topological order. In the worst case, every new arc can
invalidate the current topological order and trigger a search that visits a large
part of the graph, so the method does not improve the $\bigO(m^2)$ worst-case bound for
cycle detection.  But it is the starting point for asymptotic improvement.

To do better we use the topological order to limit the searching.  The search for $v$
from $w$ need not visit vertices larger than $v$ in the current order, since no such
vertex, nor any vertex reachable from such a vertex, can be $v$.  Here is the
resulting method in detail.  When a new arc $(v, w)$ has $v > w$, search for $v$		
from $w$ by calling \limitedsearch{$v$,$w$}, where the function \limitedsearch
is defined in Figure~\ref{alg:lim-search}.  In this and later functions and
procedures, a minus sign denotes set subtraction.

\begin{figure}
\setlength{\algomargin}{9em}
\begin{function}[H]
arc {\bf function} \limitedsearch{{\rm vertex} $v$, {\rm vertex} $w$}\;
\Indp
$F = \{w\}$; $A = \{(w, x) | (w, x) \text{ is an arc}\}$\;
\While {$A \ne \es$} {
	choose $(x, y) \in A$; $A = A \sm \{(x, y)\}$\;
	\lIf {$y = v$} {\Return $(x, y)$}\;
	\ElseIf {$y < v$ {\rm and} $y \not\in F$} {
		$F = F \cup \{y\}$; $A = A \cup	\{(y, z)| (y,z) \text{ is an arc}\}$\;
	}
}
\Return $\nlv$\;
\end{function}
\caption{Implementation of limited search.}
\label{alg:lim-search}
\end{figure}

In \limitedsearch, $F$ is the set of vertices visited by the search, and $A$
is the set of arcs to be traversed by the search.  An iteration of
the while loop that deletes an arc $(x, y)$ from $A$ does a {\em traversal} of
$(x, y)$.  The choice of which arc in $A$ to traverse is arbitrary.  If the
addition of $(v, w)$ creates a cycle, \limitedsearch{$v$,$w$} returns an arc
$(x, y) \ne (v, w)$ on such a cycle; otherwise, it returns null.  If it returns
null, restore topological order by moving all vertices in $F$ just after $v$
(and before the first vertex following $v$, if any).  Order the vertices within
$F$ topologically, for example by making the search depth-first and ordering
the vertices in $F$ in reverse postorder with respect to the search. 
Figure~\ref{fig:lim-search} shows an example of limited search and reordering.

\begin{figure}
\centering
\includegraphics[scale=0.45]{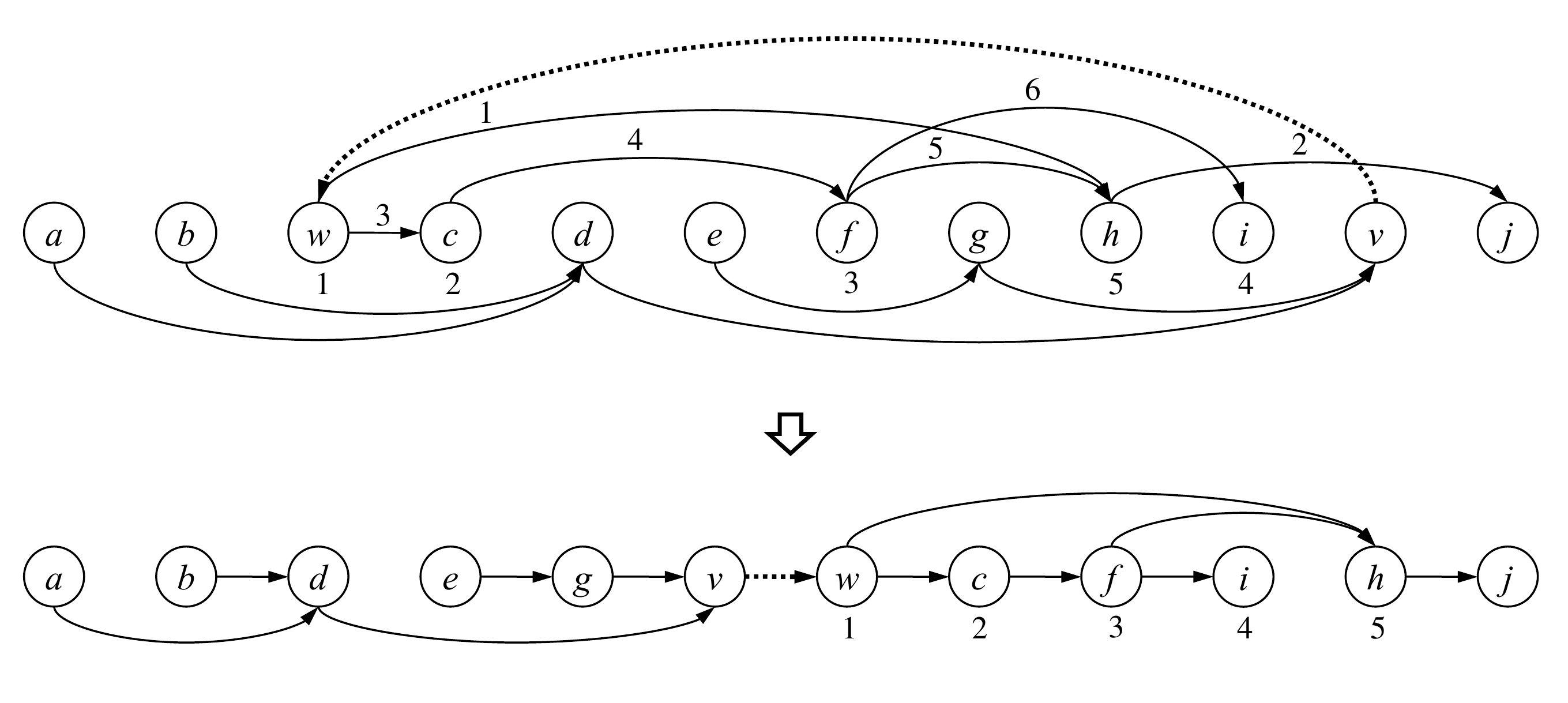}
\caption{Limited search followed by vertex reordering.  Initial topological
order is left-to-right.  Arcs are numbered in order of traversal; the search is
depth-first.  Visited vertices are $w$, $c$, $f$, $h$, $i$, $j$.  They are
numbered in reverse postorder with respect to the search and reordered
correspondingly.}
\label{fig:lim-search}
\end{figure}

Before discussing how to implement the reordering, we bound the total time for
the limited searches.  If we represent $F$ and $A$ as linked lists and mark
vertices as they are added to $F$, the time for a search is $\bigO(1)$ plus
$\bigO(1)$ per arc traversal. Only the last search, which does at most $m$ arc
traversals, can report a cycle.  To bound the total number of arc traversals,
we introduce the notion of {\em relatedness}.  We define a vertex and an arc to
be {\em related} if some path contains both the vertex and the arc, and {\em
unrelated} otherwise.  This definition does not depend on whether the vertex or
the arc occurs first on the path; they are related in either case.  If the
graph is acyclic, only one order is possible, but in a cyclic graph, a vertex
can occur before an arc on one path and after the arc on a different path.  If
either case occurs, or both, the vertex and the arc are related.

\begin{lemma}
\label{lem:lim-search-rel}
Suppose the addition of $(v, w)$ does not create a cycle but does trigger a search.
Let $(x, y)$ be an arc traversed during the (unsuccessful) search for $v$ from $w$.  Then $v$
and $(x, y)$ are unrelated before the addition but related after it.

\end{lemma}

\begin{proof}
Let $<$ be the topological order before the addition of $(v, w)$.  Since $x < v$, for
$v$ and $(x, y)$ to be related before the addition there must be a path containing
$(x, y)$ followed by $v$.  But then there is a path from $x$ to $v$.  Since there is
a path from $w$ to $x$, the addition of $(v, w)$ creates a cycle, a contradiction.
Thus $v$ and $(x, y)$ are unrelated before the addition.  After the addition there is a
path from $v$ through $(v, w)$ to $(x, y)$, so $v$ and $(x, y)$ are related. \qed

\end{proof}

The number of related vertex-arc pairs is at most $nm$, so the number of arc
traversals during all limited searches, including the last one, is at most $nm + m$.
Thus the total search time is $\bigO(nm)$.

Shmueli \cite{Shmueli1983} suggested this method but did not analyze it.  Nor did he
give an efficient way to do the reordering; he merely hinted that one could modify
his numbering scheme to accomplish this.  According to Shmueli, ``This may force us
to use real numbers (not a major problem)."  In fact, it {\em is} a major problem,
because the precision required may be unrealistically high.

To do the reordering efficiently, we need a representation more complicated than a
simple numbering scheme.  We use instead a solution to the {\em dynamic ordered list}
problem: represent a list of distinct elements so that order queries (does $x$ occur
before $y$ in the list?), deletions, and insertions (insert a given non-list element
just before, or just after, a given list element) are fast.  Solving this problem is
tantamount to addressing the precision question that Shmueli overlooked.  Dietz and
Sleator \cite{Dietz1987} gave two related solutions.  Each takes $\bigO(1)$ time
worst-case for an order query or a deletion.  For an insertion, the first takes
$\bigO(1)$ amortized time; the second, $\bigO(1)$ time worst-case.  Bender et
al. \cite{Bender2002} simplified the Dietz-Sleator methods.  With any of these
methods, the time for reordering after an arc addition is bounded by a constant
factor times the search time, so $m$ arc additions take $\bigO(nm)$ time.

There is a simpler way to do the reordering, but it requires rearranging all {\em
affected vertices}, those between $w$ and $v$ in the order (inclusive): move all
vertices visited by the search after all other affected vertices, preserving the
original order within each of these two sets.  Figure~\ref{fig:lim-search-alt}
illustrates this alternative reordering method.  We call a topological ordering
algorithm {\em local} if it reorders only affected vertices.  Except for
Shmueli's unlimited search algorithm and the recent algorithm of Bender et
al.~\cite{Bender2009}, all the algorithms we discuss are local.

\begin{figure}
\centering
\includegraphics[scale=0.45]{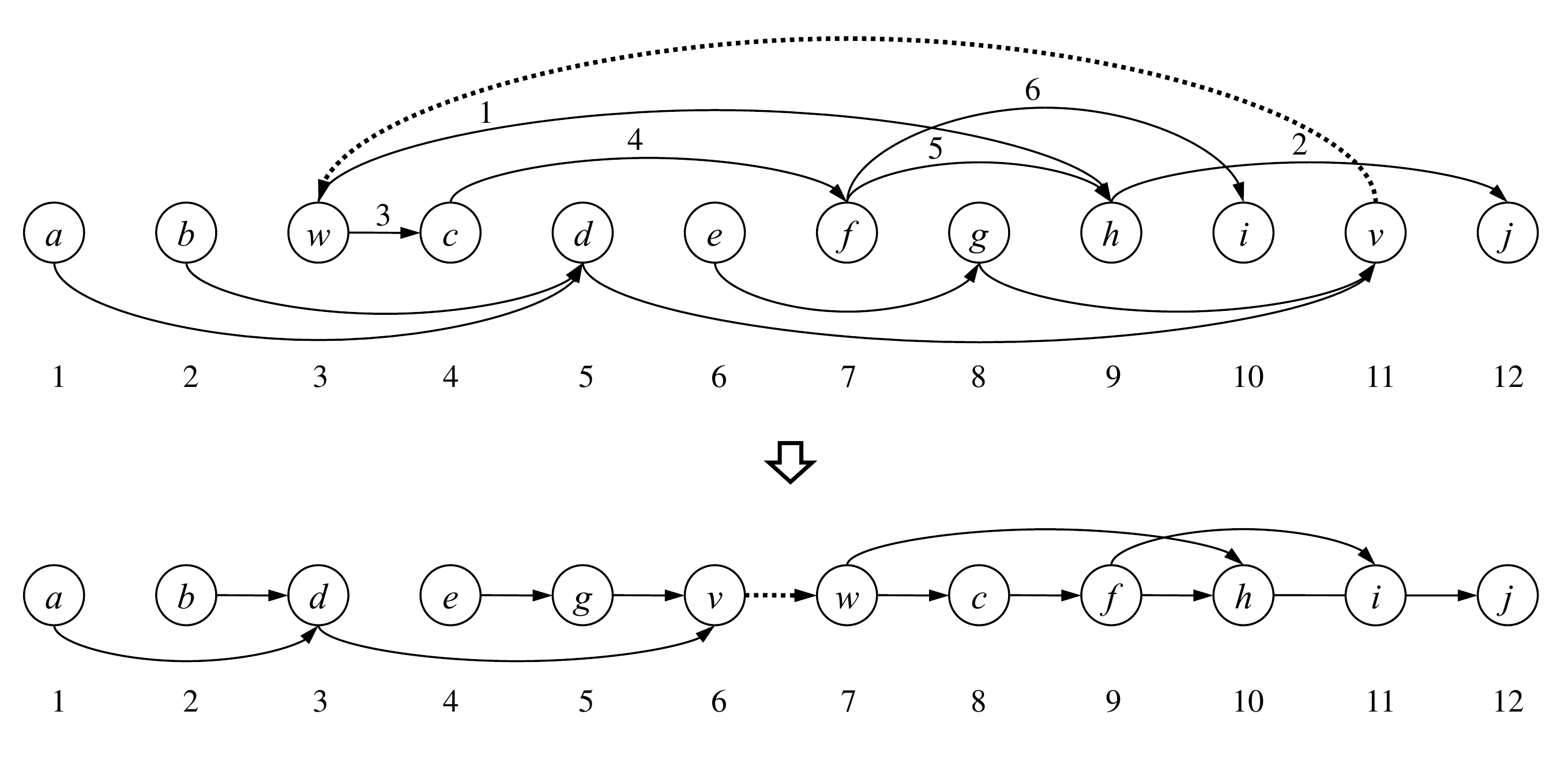}
\caption{Alternative method of restoring topological order after a limited
search of the graph in Figure~\ref{fig:lim-search}.  The vertices are numbered
in topological order. The affected vertices are
$w$,$c$,$d$,$e$,$f$,$g$,$h$,$i$,$v$. Arcs are numbered in order of traversal. 
The affected vertices are reordered by moving the visited vertices
$w$,$c$,$f$,$h$,$i$ after the unvisited vertices $d$,$e$,$g$,$v$.}
\label{fig:lim-search-alt}
\end{figure}

We can do this reordering efficiently even if the topological order is
explicitly represented by a one-to-one mapping between the vertices and the
integers from $1$ through $n$.  This makes the method a topological sorting
algorithm as defined in Section \ref{sec:intro}. This method was proposed and
analyzed by Marchetti-Spaccamela et al. \cite{Marchetti1996}.  The reordering
time is $\bigO(n)$ per arc addition; the total time for $m$ arc additions is
$\bigO(nm)$.

\section{Two-Way Search} \label{sec:2way-search}

\SetKwFunction{compatiblesearch}{Compatible-Search}
\SetKwFunction{vertexguidedsearch}{Vertex-Guided-Search}
\SetKwFunction{searchstep}{Search-Step}

We can further improve cycle detection and topological ordering by making the search
two-way instead of one-way: when a new arc $(v, w)$ has $v > w$, concurrently search
forward from $w$ and backward from $v$ until some vertex is reached from both
directions (there is a cycle), or enough arcs are traversed to guarantee that the
graph remains acyclic; if so, rearrange the visited vertices to restore topological
order.

Each step of the two-way search traverses one arc $(u, x)$ forward and one arc $(y,
z)$ backward.  To make the search efficient, we make sure that these arcs are {\em
compatible}, by which we mean that $u < z$ (in the topological order before $(v, w)$
is added).  Here is the resulting method in detail.  For ease of notation we adopt
the convention that the minimum of an empty set is bigger than any other value and
the maximum of an empty set is smaller than any other value.  Every vertex is in one
of three states: {\em unvisited}, {\em forward} (first visited by the forward
search), or {\em backward} (first visited by the backward search).  Before any arcs
are added, all vertices are unvisited.  The search maintains the set $F$ of
forward vertices and the set $B$ of backward vertices: if the search does not
detect a cycle, certain vertices in $B \cup F$ must be reordered to restore
topological order.  The search also maintains the set $A_F$ of arcs to be
traversed forward and the set $A_B$ of arcs to be traversed backward.  If the
search detects a cycle, it returns an arc other than $(v,w)$ on the cycle; if
there is no cycle, the search returns null.

When a new arc $(v, w)$ has $v > w$, search forward from $w$ and backward from
$v$ by calling \compatiblesearch{$v$,$w$}, where the function \compatiblesearch
is defined in Figure~\ref{alg:comp-search}.

\begin{figure}
\setlength{\algomargin}{4.3em}
\begin{function}[H]
arc {\bf function} \compatiblesearch{{\rm vertex} $v$, {\rm vertex} $w$}\;
\Indp
    $F = \{w\}$; $B = \{v\}$; $A_F = \{(w, x)|(w,x) \text{ is an arc}\}$; $A_B =
    \{(y, v)|(y,v) \text{ is an arc}\}$\; \While {$\exists (u,x) \in A_F,\exists
    (y,z)\in A_B \hspace{2pt}(u < z)$} { choose $(u, x) \in A_F$ and $(y, z) \in A_B$ with $u < z$\;
        $A_F = A_F \sm \{(u, x)\}$; $A_B = A_B \sm \{(y, z)\}$\;
        \lIf {$x \in B$} {\Return $(u, x)$} \lElseIf {$y \in F$} {\Return $(y,
        z)$}\; 
        \If {$x \not\in F$} {
        	$F = F \cup \{x\}$; $A_F = A_F \cup \{(x,q)|(x,q) \text{ is an
        	arc}\}$\; }
        
        \If {$y \not\in B$} {
        	$B = B \cup \{y\}$; $A_B = A_B \cup \{(r,y)|(r,y) \text{ is an
        	arc}\}$\; }
    }
	\Return $\nlv$
\end{function}
\caption{Implementation of compatible search.}
\label{alg:comp-search}
\end{figure}

In compatible search, an iteration of the while loop is a {\em search step}. 
The step does a {\em forward traversal} of the arc $(u, x)$ that it deletes from
$A_F$ and a {\em backward traversal} of the arc $(y, z)$ that it deletes from $A_B$. 
The choice of which pair of arcs to traverse is arbitrary, as long as they are
compatible.  If the addition of $(v, w)$ creates a cycle, it is possible for a
single arc $(u, z)$ to be added to both $A_F$ (when $u$ becomes forward) and to
$A_B$ (when $z$ becomes backward). It is even possible for such an arc to be
traversed both forward and backward in the same search step, but if this happens
it is the last search step.  Such a double traversal does not affect the
correctness of the algorithm.  Unlike limited search, compatible search can
visit unaffected vertices (those less than $w$ or greater than $v$ in
topological order), but this does not affect correctness, only efficiency.  If
the search returns null, restore topological order as follows.  Let $t =
\min(\{v\} \cup \{u| \exists (u, x) \in A_F \})$.  Let $F_< = \{x \in F | x < t
\}$ and $B_> = \{y \in B | y > t\}$.  If $t = v$, reorder as in limited search
(Section~\ref{sec:lim-search}): move all vertices in $F_<$ just after $t$.  (In
this case $B_> = \es$.)  Otherwise $(t < v)$, move all vertices in $F_<$ just
before $t$ and all vertices in $B_>$ just before all vertices in $F_<$.  In
either case, order the vertices within $F_<$ and within $B_>$ topologically. 
Figure~\ref{fig:comp-search} illustrates compatible search and reordering. 

\begin{figure}
\centering
\includegraphics[scale=0.45]{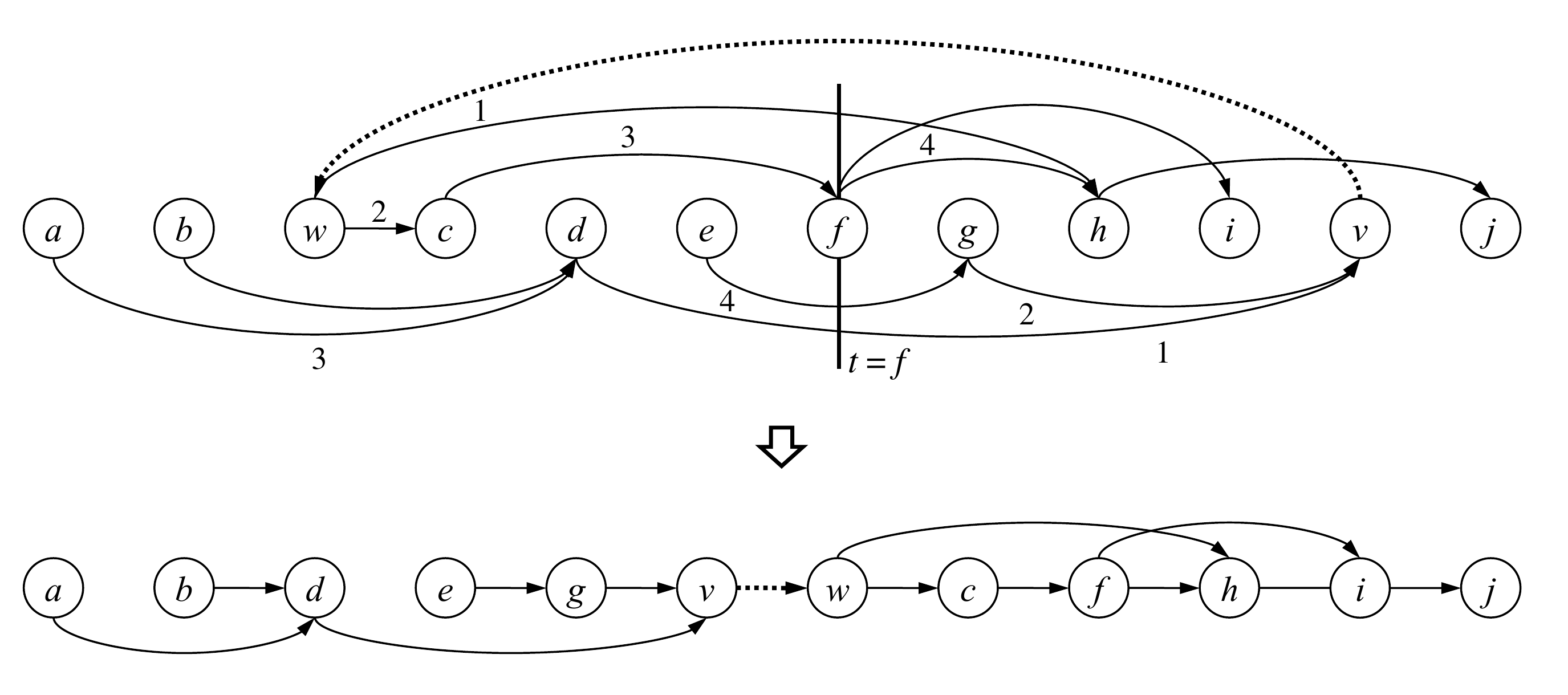}
\caption{
Compatible search of the graph in Figure~\ref{fig:lim-search} and restoration of
topological order. Traversed arc pairs are numbered in order of traversal. 
Forward vertices are $w$,$c$,$f$,$h$,$i$; backward vertices are
$v$,$d$,$g$,$e$,$a$.  After the search, $A_F = \{(f,i), (h,j)\}$; $t = f$; $F_<
= \{w,c\}$; $B_> = \{v,g\}$.  Reordering moves vertices in $F_<$ just before
$t$ and all vertices in $B_>$ just before those in $F_<$, arranging each
internally in topological order.}
\label{fig:comp-search}
\end{figure}

\begin{theorem}
\label{thm:2way-corr}
Compatible search correctly detects cycles and maintains a topological order.

\end{theorem}

\begin{proof}
The algorithm maintains the invariant that every forward vertex is reachable from $w$
and $v$ is reachable from every backward vertex.  Thus if $(u, x)$ with $x \in
B$ is traversed forward, there is a cycle consisting of a path from $w$ to $u$,
the arc $(u, x)$, a path from $x$ to $v$, and the arc $(v, w)$.  Symmetrically,
if $(y,z)$ with $y \in F$ is traversed backward, there is a cycle.  Thus if the
algorithm reports a cycle, there is one.

Suppose the addition of $(v, w)$ creates a cycle.  Such a cycle consists of a
pre-existing path $P$ from $w$ to $v$ and the arc $(v, w)$.  The existence of $P$
implies that $v > w$, so the addition of $(v, w)$ will trigger a search.  The
search maintains the invariant that either there are distinct arcs $(u, x)$ and
$(y, z)$ on $P$ with $x \le y$, $(u, x)$ is in $A_F$, and $(y, z)$ is in $A_B$,
or there is an arc $(u, z)$ in both $A_F$ and $A_B$.  In either case there is
a compatible arc pair, so the search can only stop by returning a non-null
arc.  Thus if there is a cycle the algorithm will report one.

It remains to show that if $v > w$ and the addition of $(v, w)$ does not create a cycle,
then the algorithm restores topological order.  This is a case analysis.  First
consider $(v, w)$.  If $t = v$, then $w$ is in $F_<$.  If $t < v$, then $v$ is in $B_>$ and
$w$ is in $\{t\} \cup F_<$.  In either case, $v$ precedes $w$ after the reordering.

Second, consider an arc $(x, y)$ other than $(v, w)$.  Before the reordering $x < y$; we
must show that the reordering does not reverse this order.  There are five cases:

Case 1: neither $x$ nor $y$ is in $F_< \cup B_>$.  Neither $x$ nor $y$ is reordered.

Case 2: $x$ is in $F_<$.  Vertex $y$ must be forward.  If $y < t$ then $y$ is in $F_<$, and the
order of $x$ and $y$ is preserved because the reordering within $F_<$ is topological.  If $y
= t$, then $t < v$, so the reordering inserts $x$ before $t = y$.  If $y > t$, the reordering
does not move $y$ and inserts $x$ before $y$.

Case 3: $y$ is in $F_<$ but $x$ is not.  Vertex $x$ is not moved, and $y$ follows $x$ after the
reordering since vertices in $F_<$ are only moved higher in the order.

Case 4: $y$ is in $B_>$.  Vertex $x$ must be backward.  Then $x \ne t$, since $x = t$ would
imply $t = v$ (since $x$ is backward) and $y > v$, which is impossible.  If $x > t$ then $x$ is
in $B_>$, and the order of $x$ and $y$ is preserved because the reordering within $B_>$ is
topological.  If $x < t$, the reordering does not move $x$ and inserts $y$ after $x$.

Case 5: $x$ is in $B_>$ but $y$ is not. Vertex $y$ is not moved, and $y$ follows $x$ after the
reordering since vertices in $B_>$ are only moved lower in the order.

We conclude that the reordering restores topological order.  \qed

\end{proof}

A number of implementation details remain to be filled in.  Before doing this, we
prove the key result that bounds the efficiency of two-way compatible search: the
total number of arc traversals over $m$ arc additions is $\bigO(m^{3/2})$.  To
prove this, we extend the notion of relatedness used in Section~\ref{sec:lim-search} to arc
pairs: two distinct arcs are {\em related} if they are on a common path.  Relatedness
is symmetric: the order in which the arcs occur on the common path is irrelevant.
(In an acyclic graph only one order is possible, but in a graph with cycles both
orders can occur, on different paths.)  The following lemma is analogous to
Lemma~\ref{lem:lim-search-rel}:

\begin{lemma}
\label{lem:2way-search-rel}
Suppose the addition of $(v, w)$ triggers a search but does not create a cycle.  Let
$(u, x)$ and $(y, z)$, respectively, be compatible arcs traversed forward and backward
during the search, not necessarily during the same search step.  Then $(u, x)$ and $(y,
z)$ are unrelated before the addition of $(v, w)$ but are related after the addition.

\end{lemma}

\begin{proof}
Since adding $(v, w)$ does not create a cycle, $(u, x)$ and $(y, z)$ must be
distinct.  Suppose $(u, x)$ and $(y, z)$ were related before the addition of $(v,
w)$.  Let $P$ be a path containing both.  The definition of compatibility is $u < z$.
But $u < z$ implies that $(u, x)$ precedes $(y, z)$ on $P$.  Since $u$ is forward and
$z$ is backward, the addition of $(v, w)$ creates a cycle, consisting of a path from
$w$ to $u$, the part of $P$ from $u$ to $z$, a path from $z$ to $v$, and the arc $(v,
w)$.  This contradicts the hypothesis of the lemma.  Thus $(u, x)$ and $(y, z)$ are
unrelated before the addition of $(v, w)$.

After the addition of $(v, w)$, there is a path containing both $(u, x)$ and $(y,
z)$, consisting of $(y, z)$, a path from $z$ to $v$, the arc $(v, w)$, a path from
$w$ to $u$, and the arc $(u, x)$.  Thus $(u, x)$ and $(y, z)$ are related after the
addition.  \qed

\end{proof}

\begin{theorem}
\label{thm:2way-search-arcs}
Over $m$ arc additions, two-way compatible search does at most $4m^{3/2} + m+1$
arc traversals.

\end{theorem}

\begin{proof}
Only the last arc addition can create a cycle; the corresponding search does at most
$m + 1$ arc traversals. (One arc may be traversed twice.)  Consider any search
other than the last. Let $A$ be the set of arcs traversed forward during the search.  Let $k$ be the number of arcs in $A$.
Each arc $(u, x)$ in $A$ has a distinct {\em twin} $(y, z)$ that was traversed
backward during the search step that traversed $(u, x)$.  These twins are compatible;
that is, $u < z$.  Order the arcs $(u, x)$ in $A$ in non-decreasing order on
$u$. Each arc $(u, x)$ in $A$ is compatible not only with its own twin but also with the
twin of each arc $(q, r)$ following $(u, x)$ in the order within $A$, because if $(y,
z)$ is the twin of $(q, r)$, $u \le q < z$.  Thus if $(u, x)$ is $i^{\text{th}}$ in
the order within $A$, $(u, x)$ is compatible with at least $k - i + 1$ twins of
arcs in $A$.  By Lemma~\ref{lem:2way-search-rel}, each such compatible pair is unrelated
before the addition of $(v, w)$ but is related after the addition.  Summing over all
arcs in $A$, we find that the addition of $(v, w)$ increases the number of related
arc pairs by at least $k(k + 1)/2$.

Call a search other than the last one {\em small} if it does no more than $2m^{1/2}$ arc
traversals and {\em big} otherwise.  Since there are at most $m$ small searches, together
they do at most $2m^{3/2}$ arc traversals.  A big search that does $2k$ arc traversals is
triggered by an arc addition that increases the number of related arc pairs by at
least $k(k + 1)/2 > km^{1/2}/2$.  Since there are at most ${m \choose 2} <
m^2/2$ related arc pairs, the total number of arc traversals during big searches is at most
$2m^{3/2}$.  \qed

\end{proof}

The example in Figure~\ref{fig:comp-search} illustrates the argument in the
proof of Theorem~\ref{thm:2way-search-arcs}.  The arcs traversed forward,
arranged in non-decreasing order by first vertex, are $(w, h)$ with twin $(d,
v)$, $(w, c)$ with twin $(g, v)$, $(c, f)$ with twin $(a, d)$, and $(f, h)$ with
twin $(e, g)$.  Arc $(w, h)$ is compatible with the twins of all arcs in $A$,
$(w, c)$ is compatible with its own twin and those of $(c, f)$ and $(f, h)$,
$(c, f)$ is compatible with its own twin and that of $(f, h)$, and $(f, h)$ is
compatible with its own twin.  There can be other compatible pairs, and indeed
there are in this example, but the proof does not use them.

Our goal now is to implement two-way compatible search so that the time per arc
addition is $\bigO(1)$ plus $\bigO(1)$ per arc traversal.  By
Theorem~\ref{thm:2way-search-arcs}, this would give a time bound of $\bigO(m^{3/2})$
for $m$ arc additions.  First we discuss the graph representation, then the
maintenance of the topological order, and finally (in this and the next section) the
detailed implementation of the search algorithm.

We represent the graph using forward and backward incidence lists: each vertex has a
list of its outgoing arcs and a list of its incoming arcs, which we call the {\em
outgoing list} and {\em incoming list}, respectively.  Singly linked lists
suffice.  We denote by $\foutl(x)$ and $\finl(x)$ the first arc on the outgoing
list and the first arc on the incoming list of vertex $x$, respectively.  We
denote by $\noutl((x, y))$ and $\ninl((x, y))$ the arcs after $(x, y)$ on the
outgoing list of $x$ and the incoming list of $y$, respectively.  In each case,
if there is no such arc, the value is null.  Adding a new arc $(v, w)$ to this
representation takes $\bigO(1)$ time.  If the addition of an arc $(v, w)$
triggers a search, we can update the graph representation either before or
after the search: arc $(v, w)$ will never be added to either $A_F$ or $A_B$.

We represent the topological order by a dynamic ordered list.  (See
Section~\ref{sec:lim-search}.)  If adding $(v, w)$ leaves the graph acyclic but
triggers a search, we reorder the vertices after the search as follows.  Determine $t$.
Determine the sets $F_<$ and $B_>$.  Determine the subgraphs induced by the vertices in
$F_<$ and $B_>$.  Topologically sort these subgraphs using either of the two linear-time
static methods (repeated deletion of sources or depth-first search).  Move the
vertices in $F_<$ and $B_>$ to their new positions using dynamic ordered list deletions
and insertions.  The number of vertices in $F \cup B$ is at most two plus the number of
arcs traversed by the search.  Furthermore, all arcs out of $F_<$ and all arcs into $B_>$
are traversed by the search.  It follows that the time for the topological
sort and reordering is at most linear in one plus the number of arcs traversed,
not including the time to determine $t$.  We discuss how to determine $t$ after
presenting some of the details of the search implementation.

We want the time of a search to be $\bigO(1)$ plus $\bigO(1)$ per arc traversal.
There are three tasks that are hard to implement in $\bigO(1)$ time: (1) adding arcs
to $A_F$ and $A_B$ (the number of arcs added as the result of an arc traversal may
not be $\bigO(1)$), (2) testing whether to continue the search, and (3) finding
a compatible pair of arcs to traverse.

By making the search vertex-guided instead of arc-guided, we simplify all of these
tasks, as well as the determination of $t$.  We do not maintain $A_F$ and $A_B$
explicitly.  Instead we partition $F$ and $B$ into {\em live} and {\em dead}
vertices.  A vertex in $F$ is live if it has at least one outgoing untraversed
arc; a vertex in $B$ is live if it has at least one incoming untraversed arc;
all vertices in $F \cup B$ that are not live are dead.  For each
vertex $x$ in $F$ we maintain a {\em forward pointer} $\outl(x)$ to the first
untraversed arc on its outgoing list, and for each vertex $y$ in $B$ we maintain
a {\em backward pointer} $\inl(y)$ to the first untraversed arc on its incoming
list; each such pointer is null if there are no untraversed arcs.  We also
maintain the sets $F_L$ and $B_L$ of live vertices in $F$ and $B$, respectively.
When choosing arcs to traverse, we always choose a forward arc indicated by a
forward pointer and a backward arc indicated by a backward pointer. The test
whether to continue the search becomes ``$\min F_L < \max B_L$.''

When a new arc $(v,w)$ has $v > w$, do the search by calling\hspace{4pt}
\vertexguidedsearch{$v$,$w$}, where the function \vertexguidedsearch is defined
in Figure~\ref{alg:vertex-search}.  It uses an auxiliary macro \searchstep, defined in Figure~\ref{alg:search-step}, intended to be expanded in-line; each
return from \searchstep returns from \vertexguidedsearch as well.  If
\vertexguidedsearch{$v$,$w$} returns null, let $t = \min(\{v\} \cup \{x \in
F | \outl(x) \ne null\}$ and reorder the vertices in $F_<$ and $B_>$ as discussed above.

\begin{figure}
\setlength{\algomargin}{7.5em}
\begin{function}[H]
arc {\bf function} \vertexguidedsearch{{\rm vertex} $v$, {\rm vertex} $w$}\;
\Indp
    $F = \{w\}; $B = \{v\}; $\outl(w) = \foutl(w)$; $\inl(v) =
    \finl(v)$\; 
    \lIf {$\outl(w) = \nlv$} {$F_L = \es$} \lElse {$F_L = \{w\}$}\;
    \lIf {$\inl(v) = \nlv$} {$B_L = \es$} \lElse {$B_L = \{v\}$}\;
    \While {$\min F_L < \max B_L$} {
        choose $u \in F_L$ and $z \in B_L$ with $u < z$; \searchstep{u, z}\;
    }
    \Return $\nlv$\;
\end{function}
\caption{Implementation of vertex-guided search.}
\label{alg:vertex-search}
\end{figure}

\begin{figure}
\setlength{\algomargin}{5.5em}
\begin{procedure}[H]
{\bf macro} \searchstep{{\rm vertex} $u$, {\rm vertex} $z$}\;
\Indp
    $(u, x) = \outl(u)$; $(y, z) = \inl(z)$\; 
    $\outl(u) = \noutl((u,x))$; $\inl(z) = \ninl((y, z))$\;
    \lIf {$\outl(u) = \nlv$} {$F_L = F_L \sm \{u\}$}; \lIf {$\inl(z) = \nlv$}
    {$B_L = B_L \sm \{z\}$}\; 
    \lIf {$x \in B$} {\Return $(u, x)$} \lElseIf {$y \in F$} {\Return $(y,z)$}\;
    \If {$x \not\in F$} {
    	$F = F \cup \{x\}$; $\outl(x) = \foutl(x)$\;
    	\lIf {$\outl(x) \ne \nlv$} {$F_L = F_L \cup \{x\}$}\;
    }
   	\If {$y \not\in B$} {
   		$B = B \cup \{y\}$; $\inl(y) = \finl(y)$\;
   		\lIf {$\inl(y) \ne \nlv$} {$B_L = B_L \cup \{y\}$}
   	}
\end{procedure}
\caption{Implementation of a search step.}
\label{alg:search-step}
\end{figure}

If we represent $F$ and $B$ by singly linked lists and $F_L$ and $B_L$ by doubly
linked lists (so that deletion takes $\bigO(1)$ time), plus flag bits for each
vertex indicating whether it is in $F$ and/or $B$, then the time for a search
step is $\bigO(1)$. The time to determine $t$ and to reorder the vertices is at
most $\bigO(1)$ plus $\bigO(1)$ per arc traversal.

It remains to implement tasks (2) and (3): testing whether to continue the
search and finding a compatible pair of arcs to traverse.  In vertex-guided
search these tasks are related: it suffices to test whether $\min F_L < \max
B_L$; and, if so, to find $u \in F_L$ and $z \in B_L$ with $u < z$.  The
historical solution is to store $F_L$ and $B_L$ in heaps (priority queues),
$F_L$ in a min-heap and $B_L$ in a max-heap, and in each iteration of the
while loop to choose $u = \min F_L$ and $z = \max B_L$.  This guarantees that $u
< z$, since otherwise the continuation test for the search would have failed.
With an appropriate heap implementation, the test $\min F_L < \max B_L$ takes
$\bigO(1)$ time, as does choosing $u$ and $z$.  Each insertion into a heap takes
$\bigO(1)$ time as well, but each deletion from a heap takes $\bigO(\log n)$
time, resulting in an $\bigO(\log n)$ time bound per search step and an
$\bigO(m^{3/2}\log n)$ time bound for $m$ arc additions.

This method is in essence the algorithm of Alpern et al.~\cite{Alpern1990}, although
their algorithm does not strictly alternate forward and backward arc traversals, and
they did not obtain a good total time bound.  Using heaps but relaxing the
alternation of forward and backward arc traversals gives methods with slightly better
time bounds~\cite{Alpern1990,Katriel2006,Kavitha2007}, the best bound to date being
$\bigO(m^{3/2} + nm^{1/2}\log n)$~\cite{Kavitha2007}.  One can further reduce the
running time by using a faster heap implementation, such as those of van Emde
Boas~\cite{Boas1977b,Boas1977}, Thorup~\cite{Thorup2004}, and Han and
Thorup~\cite{Han2002}.  Our goal is more ambitious: to reduce the overall running
time to $\bigO(m^{3/2})$ by eliminating the use of heaps.  This we do in the next
section.

\section{Soft-Threshold Search} \label{sec:soft-search}

\SetKwFunction{softthresholdsearch}{Soft-Threshold-Search}

To obtain a faster implementation of vertex-guided search, we exploit the flexibility
inherent in the algorithm by using a {\em soft threshold} $s$ to help choose $u$ and
$z$ in each search step.  Vertex $s$ is a forward or backward vertex, initially $v$.
We partition the sets $F_L$ and $B_L$ into {\em active} and {\em passive} vertices.
Active vertices are candidates for the current search step, passive vertices are
candidates for future search steps.  We maintain the sets $F_A$ and $F_P$, and
$B_A$ and $B_P$, of active and passive vertices in $F_L$ and $B_L$,
respectively.  All vertices in $F_P$ are greater than $s$; all vertices in $B_P$ are less than $s$; vertices
in $F_A \cup B_A$ can be on either side of $s$.  Searching continues while
$F_A \ne \es$ and $B_A \ne \es$.  The algorithm chooses $u$ from $F_A$ and $z$
from $B_A$ arbitrarily.  If $u < z$, the algorithm traverses an arc out of $u$
and an arc into $z$ and makes each newly live vertex active.  If $u > z$, the
algorithm traverses no arcs.  Instead, it makes $u$ passive if $u > s$ and
makes $z$ passive if $z < s$; $u > z$ implies that at least one of $u$ and $z$
becomes passive.  When $F_A$ or $B_A$ becomes empty, the algorithm updates $s$
and the vertex partitions, as follows.  Suppose $F_A$ is empty; the updating is
symmetric if $B_A$ is empty.  The algorithm makes all vertices in $B_P$ dead,
makes $s$ dead if it is live, chooses a new $s$ from $F_P$, and makes active
all vertices $x \in F_P$ such that $x \le s$.

Here are the details of this method, which we call {\em soft-threshold search}.
When a new arc $(v, w)$ has $v > w$, do the search by calling
\softthresholdsearch{$v$,$w$}, where the function \softthresholdsearch is
defined in Figure~\ref{alg:soft-search}, and procedure \searchstep is defined
as in Figure~\ref{alg:search-step}, but with $F_A$ and $B_A$ replacing $F_L$ and
$B_L$, respectively.  If \softthresholdsearch{$v$,$w$} returns null, let $t =
\min(\{v\} \cup \{x \in F| \outl(x) \ne \nlv\}$ and reorder the vertices in
$F_<$ and $B_>$ as discussed above.  Figure~\ref{fig:soft-search}
illustrates soft-threshold search.

\begin{figure}
\setlength{\algomargin}{5.5em}
\begin{function}[H]
arc {\bf function} \softthresholdsearch{{\rm vertex} $v$, {\rm vertex} $w$}\;
\Indp
$F = \{w\}$; $B = \{v\}$; $\outl(w) = \foutl(w)$; $\inl(v) =
\finl(v)$; $s = v$\;
\lIf {$\outl(w) = \nlv$} {$F_A = \es$} \lElse {$F_A = \{w\}$}; $F_P = \es$\;
\lIf {$\inl(v) = \nlv$} {$B_A = \es$} \lElse {$B_A = \{v\}$}; $B_P = \es$\;
\While {$F_A \ne \es$ and $B_A \ne \es$} {
 	choose $u \in F_A$ and $z \in B_A$\; 
   	\lIf {$u < z$} {\searchstep{u, z}}
   	\Else {
   		\If {$u > s$} {
   			$F_A = F_A \sm \{u\}$; $F_P = F_P \cup \{u\}$\;
   		}
   		\If {$z < s$} {
   			$B_A = B_A \sm \{z\}$; $B_P = B_P \cup \{z\}$\;
   		}
   	}
   	\If {$F_A = \es$} {
   		$B_P = \es$; $B_A = B_A \sm \{s\}$\;
    	\If {$F_P \ne \es$} {
       		choose $s \in F_P$; $F_A = \{x \in F_P| x \le s\}$; $F_P = F_P \sm
       		F_A$\;
       	}
    }
    \If {$B_A = \es$} {
    	$F_P = \es$; $F_A = F_A \sm \{s\}$\;
       	\If {$B_P \ne \es$} {
           	choose $s \in B_P$; $B_A = \{x \in B_P| x \ge s\}$; $B_P = B_P
           	\sm B_A$\;
       	}
    }
}
\Return $\nlv$\;
\end{function}
\caption{Implementation of soft-threshold search.}
\label{alg:soft-search}
\end{figure}

\begin{figure}[h!]
\centering
\includegraphics[scale=0.45]{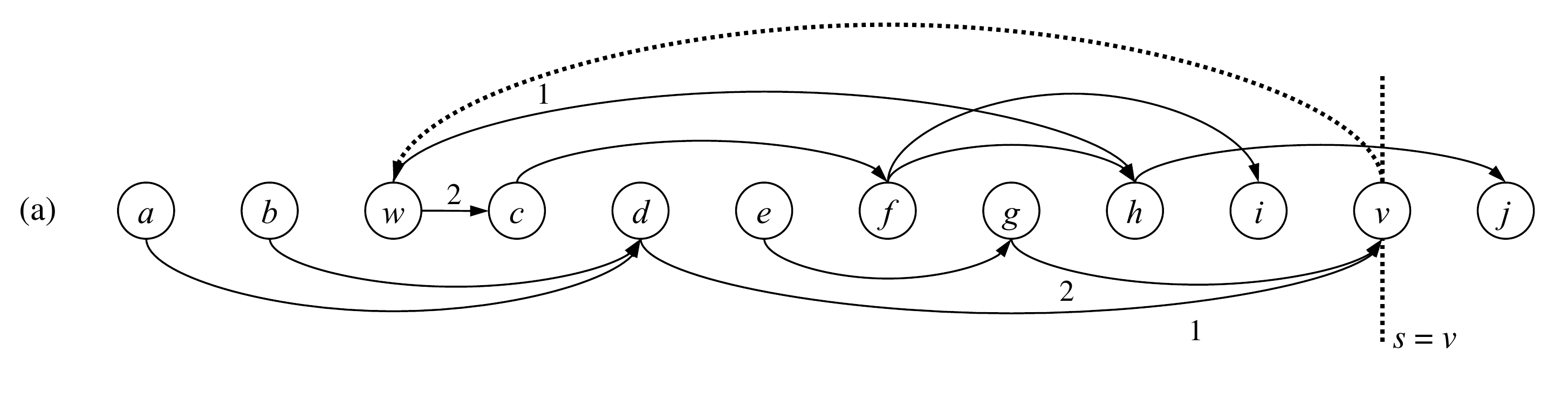}
\vspace{5pt}
\includegraphics[scale=0.45]{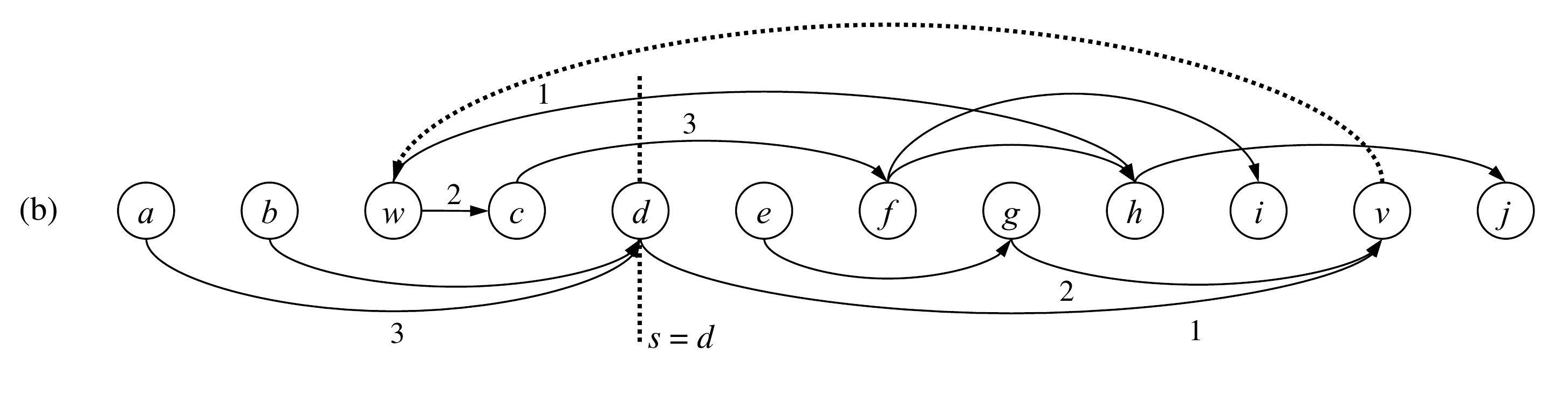}
\vspace{5pt}
\includegraphics[scale=0.45]{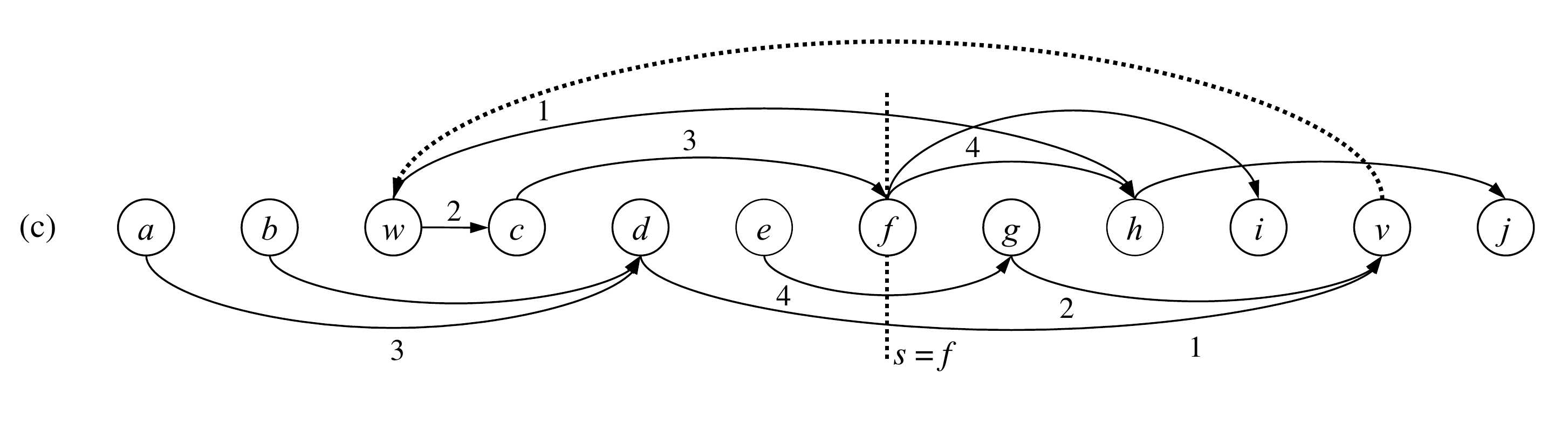}
\caption{
Soft-threshold search of the graph in Figure~\ref{fig:lim-search}.  Arc
traversal order is the same as in Figure~\ref{fig:comp-search}.  Initially $s =
v$, $F_A = \{w\}$, $B_A = \{v\}$. (a) Choosing $u = w$, $z = v$ twice causes
traversal of compatible pair $(w, h)$, $(d, v)$ followed by traversal of $(w, c)$, $(g, v)$. 
Now $F_A = \{h, c\}$, $B_A = \{d, g\}$.  Choice of $u = h$, $z = d$ moves $d$ to
$B_P$. Choice of $u = h$, $z = g$ moves $g$ to $B_P$, making $B_A$ empty.  (b)
New $s$ is $d$. Now $F_A = \{h, c\}$, $F_P = \es$, $B_A = \{d, g\}$, $B_P =
\es$. Choice of $u = h$, $z = d$ moves $h$ to $F_P$.  Choice of $u = c$, $z = d$
causes traversal of $(c, f)$, $(a, d)$, adding $f$ to $F_A$ and deleting $c$
from $F_A$.  (Vertex $a$ has no incoming arc, so it is not added to $B_A$.)
Choice of $u = f$, $z = d$ moves $f$ to $F_P$, making $F_A$ empty. (c) New $s$
is $f$. Now $F_A = \{f\}$, $F_P = \{h\}$, $B_A = \{g\}$, $B_P = \es$.  Choice
of $u = f$, $z = g$ causes traversal of $(f, h)$, $(e, g)$, deleting $g$ from
$B_A$ and making $B_A$ empty.  Since $B_P$ is also empty, the search ends.
 Reordering is the same as in Figure~\ref{fig:comp-search}. }
\label{fig:soft-search}
\end{figure}

Soft-threshold search is an implementation of vertex-guided search except that it
makes additional vertices dead, not just those with no incident arcs left to
traverse.  Once dead, a vertex stays dead.  We need to prove that this does not
affect the search outcome. First we prove that soft-threshold search terminates.

\begin{theorem}
\label{thm:soft-search-steps}
A soft-threshold search terminates after at most $n^2 + m + n$ iterations of
the while loop.

\end{theorem}

\begin{proof}
Each iteration either traverses one or two arcs or makes one or two vertices
passive. The number of times a vertex can become passive is at most the number of times it becomes
active.  Vertices become active only when they are visited (once per vertex) or when
$s$ changes.  Each time $s$ changes, the old $s$ becomes dead if it was not dead
already.  Thus $s$ changes at most $n$ times.  The number of times vertices become
active is thus at most $n + n^2$ (once per vertex visit plus once per vertex per
change in $s$).  \qed

\end{proof}

To prove correctness, we need two lemmas. 

\begin{lemma}
\label{lem:soft-search-passive}
If $x$ is a passive vertex, $x > s$ if $x$ is in $F_P$, $x < s$ if $x$ is in
$B_P$.

\end{lemma}

\begin{proof}
If $x$ is a passive vertex, $x$ satisfies the lemma when it becomes passive, and it
continues to satisfy the lemma until $s$ changes.  Suppose $x$ is forward; the
argument is symmetric if $x$ is backward. If $s$ changes because $F_A$ is empty, $x$
becomes active unless it is greater than the new $s$.  If $s$ changes because
$B_A$ is empty, $x$ becomes dead.  The lemma follows by induction on the number of search
steps.  \qed

\end{proof}

\begin{lemma}
\label{lem:soft-search-live}
Let $A_F$ be the set of untraversed arcs out of vertices in $F$, let $A_B$
be the set of untraversed arcs into vertices in $B$, let $q = \min\{u| \exists (u, x) \in
A_F\}$, and let $r = \max\{z| \exists (y, z) \in A_B\}$.  Then $q$ and $r$ remain live vertices
until $q > r$.

\end{lemma}

\begin{proof}
If $q$ and $r$ remain live vertices until $q = \infty$ or $r = -\infty$, the lemma
holds.  Thus suppose $q$ dies before $r$ and before either $q = \infty$ or $r =
-\infty$. When $q$ dies, $q = s$ or $q$ is passive, and $B_A = \es$.  Since $r$
is still live, $r$ is passive.  By Lemma~\ref{lem:soft-search-passive}, $q \ge s
> r$.  The argument is symmetric if $r$ dies before $q$ and before either $q =
\infty$ or $r = -\infty$.  \qed

\end{proof}

\begin{theorem}
Soft-threshold search is correct.

\end{theorem}

\begin{proof}
Let $q$ and $r$ be defined as in Lemma~\ref{lem:soft-search-live}.  By that
lemma, the search will continue until a cycle is detected or $q > r$. While the
search continues, it traverses arcs in exactly the same way as vertex-guided
search.  Once $q > r$, the continuation test for vertex-guided search fails. 
If the graph is still acyclic, the continuation test for soft-threshold search
may not fail immediately, but no additional arcs can be traversed; any
additional iterations of the while loop merely change the
state (active, passive, or dead) of various vertices.  Such changes do not affect the
outcome of the search.  \qed

\end{proof}

To implement soft-threshold search, we maintain $F_A$, $F_P$, $B_A$, and
$B_P$ as doubly-linked lists.  The time per search step is $\bigO(1)$, not counting the
computations associated with a change in $s$ (the two code blocks at the end of the
while loop that are executed if $F_A$ or $B_A$ is empty, respectively).

The remaining freedom in the algorithm is the choice of $s$.  The following
observation guides this choice.  Suppose $s$ changes because $F_A$ is empty.  The
algorithm chooses a new $s$ from $F_P$ and makes active all vertices in $F_P$ that
are no greater than $s$.  Consider the next change in $s$.  If this change occurs
because $F_A$ is again empty, then all the vertices that were made active by the
first change of $s$, including $s$, are dead, and hence can never become active
again.  If, on the other hand, this change occurs because $B_A$ is empty, then all
the forward vertices that remained passive after the first change in $s$ become dead,
and $s$ becomes dead if it is not dead already.  That is, either the vertices in
$F_P$ no greater than the new $s$, or the vertices in $F_P$ no less than the new $s$,
are dead after the next change in $s$. Symmetrically, if $s$ changes because
$B_A$ is empty, then either all the vertices in $B_P$ no less than the new $s$,
or all the vertices in $B_P$ no greater than the new $s$, are dead after the
next change in $s$.  To minimize the worst case, we always select $s$ to be the
median of the set of choices.  This takes time linear in the number of
choices~\cite{Blum1973,Schonhage1976,DorZ1999}.

\begin{theorem}
\label{thm:soft-search-time}
If each new $s$ is selected to be the median of the set of choices, soft-threshold
search takes $\bigO(m^{3/2})$ time over $m$ arc additions.

\end{theorem}

\begin{proof}
Consider a soft-threshold search.  For each increase in $s$ we charge an amount
equal to the number of vertices in $F_P$ when the change occurs; for each
decrease in $s$ we charge an amount equal to the number of vertices in $B_P$
when the change occurs. The charge covers the time spent in the code block
associated with the change, including the time to find the new $s$ (the median)
and the time to make vertices passive or dead, all of which is linear in the
charge. The charge also covers any time spent later to make passive any
vertices that became active as a result of the change; this time is $\bigO(1)$
for each such vertex.  The remainder of the search time is $\bigO(1)$ for
initialization plus $\bigO(1)$ per arc traversal.  We claim that the total
charge is $\bigO(1)$ per arc traversal.  The theorem follows from the claim and
Theorem~\ref{thm:2way-search-arcs}.

The number of vertices in $F \cup B$ is at most the number of arc traversals.  We
divide the total charge among these vertices, at most two units per vertex.  The
claim follows.

Consider a change in $s$ other than the last.  Suppose this is an increase.  Let $k$
be the number of vertices in $F_P$ when this change occurs; $k$ is also the charge
for the change.  Since $s$ is selected to be the median of $F_P$, at
least $\lceil k/2 \rceil$ vertices in $F_P$ are no greater than $s$, and
at least $\lceil k/2 \rceil$ vertices in $F_P$ are no less than $s$.  If the
next change in $s$ is an increase, all the vertices in $F_P$ no greater than $s$ must be dead by the time of the next change.  If the next
change in $s$ is a decrease, all the vertices in $F_P$ no less than $s$ will be made
dead by the next change, including $s$ if it is not dead already.  In either case we
associate the charge of $k$ with the at least $\lceil k/2 \rceil$ vertices that
become dead after the change in $s$ but before or during the next change in $s$.

A symmetric argument applies if $s$ decreases.  The charge for the last change in $s$
we associate with the remaining live vertices, at most one unit per vertex.  \qed

\end{proof}

Theorem~\ref{thm:soft-search-time} holds (with a bigger constant factor) if each new
$s$ is an approximate median of the set of choices; that is, if $s$ is larger
than $\epsilon k$ and smaller than $\epsilon k$ of the $k$ choices, for some
fixed $\epsilon > 0$.  An alternative randomized method is to select each new
$s$ uniformly at random from among the choices.

\begin{theorem}
\label{thm:soft-search-rtime}
If each new $s$ is chosen uniformly at random from among the set of choices,
soft-threshold search takes $\bigO(m^{3/2})$ expected time over $m$ arc additions.

\end{theorem}

\begin{proof}
Each selection of $s$ takes time linear in the number of choices.  We charge for the
changes in $s$ exactly as in the proof of Theorem~\ref{thm:soft-search-time}.  The
search time is then $\bigO(1)$ plus $\bigO(1)$ per arc traversal plus $\bigO(1)$ per unit of
charge.  We shall show that the expected total charge for a search is at most
linear in the number of vertices in $F \cup B$, which in turn is at most the
number of arc traversals.  The theorem follows from the bound on expected total
charge and Theorem~\ref{thm:2way-search-arcs}.

The analysis of the expected total charge is much like the
analysis~\cite{Knuth1972} of Hoare's ``quick select''
algorithm~\cite{Hoare1961}. We construct an appropriate recurrence and prove a
linear bound by induction.  Consider the situation just before some search
step.  Let $\E(k)$ be the maximum expected total future charge, given that at
most $k$ distinct vertices are candidates for $s$ during future changes of $s$. 
(A vertex can be a candidate more than once, but we only count it once.)  The
maximum is over the possible current states of all the data structures; the
expectation is over future choices of $s$. We prove by induction on $k$ that
$\E(k) \le 4k$.

If $s$ does not change in the future, or if the next change in $s$ is the last
one, then the total future charge is at most $k$.  Suppose the next change of $s$
is not the last, and the next choice of $s$ is from among $j$ candidates.  Each of
these $j$ candidates is selected with probability $1/j$.  If the new $s$ is the
$i^{\text{th}}$ smallest among the candidates, then at least $\min\{i, j - i + 1\}$
of these candidates cannot be future candidates.  The charge for this change in $s$
is $j$.  The maximum expected future charge, including that for this change in $s$,
is at most $j + \sum_{i=1}^{j/2} (2\E(k - i)/j)$ if $j$ is even, at most $j +
\E(k - \lceil j/2 \rceil)/j + \sum_{i=1}^{\lfloor j/2 \rfloor} (2\E(k - i)/j)$ if $j$ is
odd.  Using the induction hypothesis $\E(k') \le 4k'$ for $k' < k$, we find that the
maximum expected future charge is at most $j + \sum_{i=1}^{j/2} (8(k - i)/j) =
4k + j - \sum_{i=1}^{j/2} (8i/j) = 4k + j - (4/j)(j/2)(j/2+1) = 4k - 2$ if $j$
is even, at most $j + 4(k - \lceil j/2 \rceil)/j + \sum_{i=1}^{\lfloor j/2
\rfloor} (8(k - i)/j) = 4k + j - 4\lceil j/2 \rceil / j - \sum_{i=1}^{\lfloor j/2 \rfloor} (8i/j) = 4k +
j - (4/j)(j/2 + 1/2 + (j/2-1/2)(j/2+1/2)) < 4k + j - (4/j)(j/2)^2 = 4k$ if $j$
is odd. By induction $\E(k) \le 4k$ for all $k$.

Over the entire search, there are at most $|F \cup B|$ candidates for $s$. It
follows that the expected total charge over the entire search is at most $4|F
\cup B|$, which is at most four times the number of arcs traversed during the
search.  \qed

\end{proof}

Soft-threshold search with either method of choosing $s$ uses $\bigO(n + m)$
space, as do all the algorithms we have discussed so far.  Katriel and
Bodlaender~\cite{Katriel2006} give a set of examples on which soft-threshold
search takes $\Omga(m^{3/2})$ time no matter how $s$ is chosen, so the bounds
in Theorems~\ref{thm:soft-search-time} and~\ref{thm:soft-search-rtime} are tight.

It is natural to ask whether there is a faster algorithm.  To address this question,
we consider algorithms that maintain (at least) an explicit list of the vertices in
topological order and that do any needed reordering by moving one vertex at a time to
a new position in this list.  All known algorithms do this or can be modified to do
so with at most a constant-factor increase in running time.  We further restrict our
attention to {\em local} algorithms, those that update the order after an arc $(v,
w)$ with $v > w$ is added by reordering only affected vertices (defined in
Section~\ref{sec:lim-search}: those vertices between $w$ and $v$, inclusive).  These
vertices form an interval in the old order and must form an interval in the new
order; within the interval, any permutation is allowed as long as it restores
topological order.  Our algorithms, as well as all previous ones except for
those of Shmueli~\cite{Shmueli1983} and Bender et al.~\cite{Bender2009}, are
local. The following theorem gives a lower bound of $\Omga(n\sqrt{m})$ on the
worst-case number of vertices that must be moved by any local algorithm.  Thus
for sparse graphs ($m/n = \bigO(1)$), soft-threshold search is as fast as
possible among local algorithms.

\begin{theorem}
\label{thm:local-lb}
Any local algorithm must reorder $\Omga(n\sqrt{m})$ vertices, and hence must take
$\Omga(n\sqrt{m})$ time.

\end{theorem}

\begin{proof}
Let $p$ and $k$ be arbitrary positive integers such that $p \le k$.  We shall give an
example with $n = p(k + 1)$ vertices and $m = n - k - 1 + k(k + 1)/2$ arcs that
requires at least $pk(k + 1)/2 = nk/2$ vertex movements.  Since $p \le k$, $k(k + 1)/2 \le m
\le 3k(k + 1)/2$, so $\sqrt{m} = \Thta(k)$.  The example is such that, after $n
- k - 1$ initial arc additions, each subsequent arc addition forces at least $p$
vertices to be moved in the topological order, assuming the algorithm is local.  The total number
of vertex movements is thus at least $pk(k+1)/ 2 = \Omga(n\sqrt{m})$.  Given any
target number of vertices $n'$ and target number of arcs $m'$, we can choose $p$ and
$k$ so that $n = \Thta(n')$ and $m = \Thta(m')$, which gives the theorem.

The construction is quite simple.  Let the $n$
vertices be numbered 1 through $n$ in their original topological order.  Add $n
- k - 1$ arcs so that each interval of $p$ consecutive vertices ending in an
integer multiple of $p$ forms a path of the vertices in increasing order (so
that vertices 1 through $p$ form a path from 1 to $p$, $p + 1$ through $2p$
form a path from $p + 1$ to $2p$, and so on).  Now there are $k + 1$ paths,
each containing $p$ vertices.  Call these paths $P_1, P_2,\ldots,P_{k + 1}$, in
increasing order by first (and last) vertex.  Add an arc from the last vertex
of $P_2$ (vertex $2p$) to the first vertex of $P_1$ (vertex 1). This forms a
path from $p + 1$ through $p + 2, p + 3,\ldots$ to $2p$, then through $1, 2,
\ldots$ to $p$.  The affected vertices are the vertices 1 through $2p$, and
the only way to rearrange them to restore topological order is to move $p + 1$
through $2p$ before 1 through $p$, which takes at least $p$ individual vertex
moves. The effect is to swap $P_1$ and $P_2$ in the topological order.  Now add
an arc from the last vertex of $P_3$ to the first vertex of $P_1$.  This forces
$P_1$ to swap places with $P_3$, again requiring at least $p$ vertex moves. 
Continue adding one arc at a time in this way, forcing $P_1$ to swap places
with $P_4, P_5,\ldots, P_{k + 1}$.  After $k$ arcs additions of arcs from the
last vertex of $P_2, P_3,\ldots, P_{k + 1}$ to the first vertex of $P_1$, path
$P_1$ has been forced all the way to the top end of the topological order.  Now
ignore $P_1$ and repeat the construction with $P_2$, forcing it to move past
$P_3, P_4,\ldots, P_{k + 1}$ by adding arcs $(3p, p + 1), (4p, p + 1),\ldots,
((k + 1)p, p + 1)$.  Do the same with $P_3, P_4,\ldots, P_k$. The total number
of arcs added that force vertex moves is $k(k + 1)/2$.  Each of these added
arcs forces at least $k$ vertex moves.  Figure~\ref{fig:soft-lb} gives an
example of the construction.  \qed

\end{proof}

\begin{figure}[h!]
\centering
\includegraphics[scale=0.45]{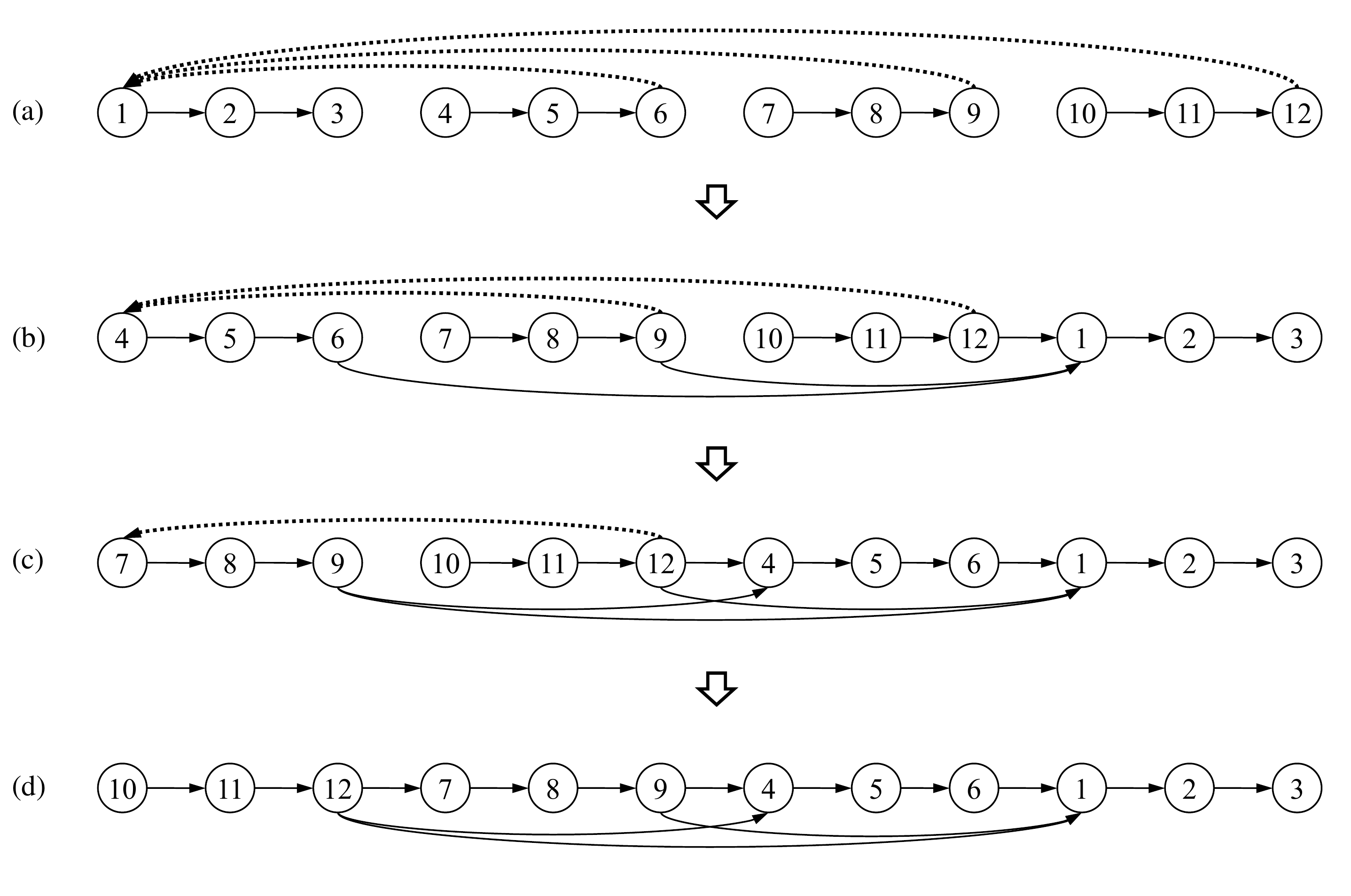}
\caption{The $\Omga(nm^{1/2})$ vertex reordering construction for $p = 3$ and $k
= 3$, yielding an example with $n = 12$ vertices and $m = 14$ arcs. (a)
 Insertion of arc $(6, 1)$ moves 1, 2, 3 past 4, 5, 6.  Insertion of $(9, 1)$
 moves 1, 2, 3 past 7, 8, 9.  Insertion of $(12, 1)$ moves 1, 2, 3 past 10, 11,
 12. (b) Insertion of $(9, 4)$ moves 4, 5, 6 past 7, 8, 9.  Insertion of $(12,
 4)$ moves 4, 5, 6 past 10, 11, 12. (c)	Insertion of $(12, 7)$ moves 7, 8, 9
 past 10, 11, 12. (d) Final order.
}
\label{fig:soft-lb}
\end{figure}

The $\Omga(n\sqrt{m})$ bound on vertex reorderings is tight.  An algorithm that
achieves this bound is a two-way search that does not alternate forward and backward
arc traversals but instead does forward arc traversals until visiting an unvisited
vertex less than $v$, then does backward arc traversals until visiting an
unvisited vertex greater than $w$, and repeats.  Each forward traversal is
along an arc $(u, x)$ with $u$ minimum; each backward traversal is along an arc
$(y, z)$ with $z$ maximum.  Searching continues until a cycle is detected or
there is no compatible pair of untraversed arcs.  If the search stops without
detecting a cycle, the algorithm reorders the vertices in the same way as in
two-way compatible search.  One can prove that this method reorders
$\bigO(n\sqrt{m})$ vertices over $m$ arc additions by counting related vertex
pairs (as defined in the next section: two vertices are related if one path
contains both). Unfortunately we do not know an implementation of this
algorithm with an overall time bound approaching the bound on vertex reorderings.

For algorithms that reorder one vertex at a time but are allowed to move
unaffected vertices, the only lower bound known is the much weaker one of
Ramalingam and Reps~\cite{Ramalingam1994}.  They showed that $n - 1$ arc
additions can force any algorithm, local or not, to do $\Omga(n\log n)$ vertex
moves.

\section{Topological Search} \label{sec:top-search}

\SetKwFunction{topologicalsearch}{Topological-Search}
\SetKwFunction{reorder}{Reorder}

Soft-threshold search is efficient on sparse graphs but becomes less and less
efficient as the graph becomes denser; indeed, if $m = \Omga(n^2)$ the time bound is
$\bigO(n^3)$, the same as that of one-way limited search (Section
\ref{sec:lim-search}). In this section we give an alternative algorithm that is efficient for dense graphs.
The algorithm uses two-way search, but differs in three ways from the methods
discussed in Sections \ref{sec:2way-search} and \ref{sec:soft-search}: it balances
vertices visited instead of arcs traversed (as in the method sketched at the end of
Section~\ref{sec:soft-search}); it searches the topological order instead of the
graph; and it uses a different reordering method, which has the side benefit of
making it a topological sorting algorithm.  We call the algorithm {\em topological
search}.

We represent the topological order by an explicit mapping between the vertices
and the integers from 1 to $n$.  We denote by $\position(v)$ the number of
vertex $v$ and by $\vertex(i)$ the vertex with number $i$.  We implement
$\vertex$ as an array.  The initial numbering is arbitrary; it is topological
since there are no arcs initially.  If $v$ and $w$ are vertices, we test $v <
w$ by comparing $\position(v)$ to $\position(w)$.  We represent the graph by an
adjacency matrix $A: A(v, w) = 1$ if $(v, w)$ is an arc, $A(v, w) = 0$ if not. 
Testing whether $(v, w)$ is an arc takes $\bigO(1)$ time, as does adding an
arc.  Direct representation of $A$ uses $\bigO(n^2)$ bits of space;
representation of $A$ by a hash table reduces the space to $\bigO(n + m)$ but
makes the algorithm randomized.
 
To simplify the running time analysis and the extension to strong component
maintenance (Section \ref{sec:strong}), we test for cycles after the search. 
Thus the algorithm consists of three parts: the search, the cycle test, and the
vertex reordering.  Let $(v, w)$ be a new arc with $v > w$.  The search
examines every affected vertex (those between $w$ and $v$ in the order,
inclusive).  It builds a queue $F$ of vertices reachable from $w$ by searching
forward from $w$, using a current position $i$, initially $\position(w)$. 
Concurrently, it builds a queue $B$ of vertices from which $v$ is
reachable by searching backward from $v$, using a current position $j$,
initially $\position(v)$. It alternates between adding a vertex to $F$
and adding a vertex to $B$ until the forward and backward searches meet.  When
adding a vertex $z$ to $F$ or $B$, the method sets $\vertex(\position(z)) =
\nlv$.

In giving the details of this method, we use the following notation for queue
operations: $\eq$ denotes an empty queue; $\inject(x, Q)$ adds element $x$ to
the back of queue $Q$; $\pop(Q)$ deletes the front element $x$ from queue $Q$
and returns $x$; if $Q$ is empty, $\pop(Q)$ leaves $Q$ empty and returns null. 
Do the search by calling \topologicalsearch{$v$,$w$}, where procedure
\topologicalsearch is defined in Figure~\ref{alg:top-search}.

\begin{figure}
\setlength{\algomargin}{7em}
\begin{procedure}[H]
{\bf procedure} \topologicalsearch{{\rm vertex} $v$, {\rm vertex} $w$}\;
\Indp
    $F = \eq$; $B = \eq$; $\inject(w, F)$; $\inject(v, B)$\;
    $i = \position(w)$; $j = \position(v)$; $\vertex(i) = \vertex(j) = \nlv$\; 
    \While {true} {
        $i = i + 1$\;
        \lWhile {$i < j$ {\rm and} $\forall u \in F (A(\position(u), i) = 0)$}
        {$i = i + 1$}\; 
        \lIf {$i = j$} {\Return} \Else {
        	$\inject(\vertex(i), F)$; $\vertex(i) = \nlv$\;
        }
        $j = j - 1$\; 
        \lWhile {$i < j$ {\rm and} $\forall z \in B(A(j, \position(z)) = 0)$}
        {$j = j - 1$}\;
        \lIf {$i = j$} {\Return} \Else {
        	$\inject(\vertex(j),B)$; $\vertex(j) = \nlv$\;
        }
	}
\end{procedure}
\caption{Implementation of topological search.}
\label{alg:top-search}
\end{figure}

Once the search finishes, test for a cycle by checking whether there is an arc
$(u, z)$ with $u$ in $F$ and $z$ in $B$.  If there is no such arc, reorder the vertices as
follows.  Let $F$ and $B$ be the queues at the end of the search, and let $k$ be
the common value of $i$ and $j$ at the end of the search.  Then $\vertex(k) =
\nlv$.  If the search stopped after incrementing $i$, then $\vertex(k)$ was
added to $B$, and $F$ and $B$ contain the same number of vertices.  Otherwise,
the search stopped after decrementing $j$, $\vertex(k)$ was added to $F$, and
$F$ contains one more vertex than $B$.  In either case, the number of positions
$g \ge k$ such that $\vertex(g) = \nlv$ is $|F|$, and the number of positions $g
< k$ such that $\vertex(g) = \nlv$ is $|B|$. Reinsert the vertices in $F \cup B$
into the vertex array, moving additional vertices as necessary, by calling
\reorder, using as the initial values of $F$, $B$, $i$, $j$ their values at the
end of the search, where procedure \reorder is defined in
Figure~\ref{alg:reorder}.

\begin{figure}
\setlength{\algomargin}{7.3em}
\begin{procedure}[H]
{\bf procedure} \reorder\;
\Indp
	\While {$F \ne \eq$} {
        \If {$\vertex(i) \ne \nlv$ {\rm and} $\exists u \in F (A(u,\vertex(i)) =
        1)$} { 
        	$\inject(\vertex(i), F)$; $\vertex(i) = \nlv$\; 
        }
        \If {$\vertex(i) = \nlv$} {
        	$x = \pop(F)$; $\vertex(i) = x$; $\position(x) = i$\;
       	}
        $i = i + 1$
	}
	\While {$B \ne \eq$} {
        $j = j - 1$\;
        \If {$\vertex(j) \ne \nlv$ {\rm and} $\exists z \in B (A(\vertex(j), z)
        = 1)$} { 
        	$\inject(\vertex(j), B)$; $\vertex(j) = \nlv$\; 
        } 
        \If {$\vertex(j) =
        \nlv$} { $y = \pop(B)$; $\vertex(j) = y$; $\position(y) = j$\; 
        }
	}
\end{procedure}
\caption{Implementation of reordering.}
\label{alg:reorder}
\end{figure}

The reordering process consists of two almost-symmetric while loops.  The first loop
reinserts the vertices in $F$ into positions $k$ and higher.  Variable $i$ is the current
position.  If $\vertex(i)$ is a vertex $q$ with an arc from a vertex currently in $F$,
vertex $q$ is added to the back of $F$ and $\vertex(i)$ becomes null: vertex $q$ must be
moved to a higher position.  If $\vertex(i)$ becomes null, or if $\vertex(i)$ was already
null, the front vertex in $F$ is deleted from $F$ and becomes $\vertex(i)$.  The
second loop reinserts the vertices in $B$ into positions $k - 1$ and lower in symmetric
fashion.  The only difference between the loops is that the forward loop increments $i$
last, whereas the backward loop decrements $j$ first, to avoid examining $\vertex(k)$.
The forward and backward loops are completely independent and can be executed in
parallel.  (This is not true of the forward and backward searches.)
Figure~\ref{fig:top-search} gives an example of topological search and
reordering.

\begin{figure}
\centering
\includegraphics[scale=0.45]{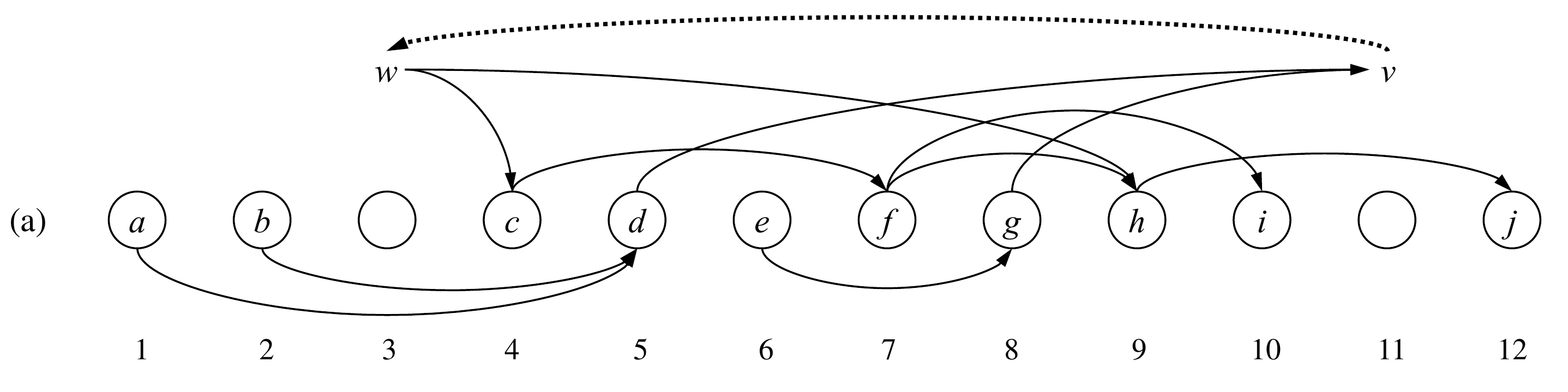}\\
\vspace{5pt}
\includegraphics[scale=0.45]{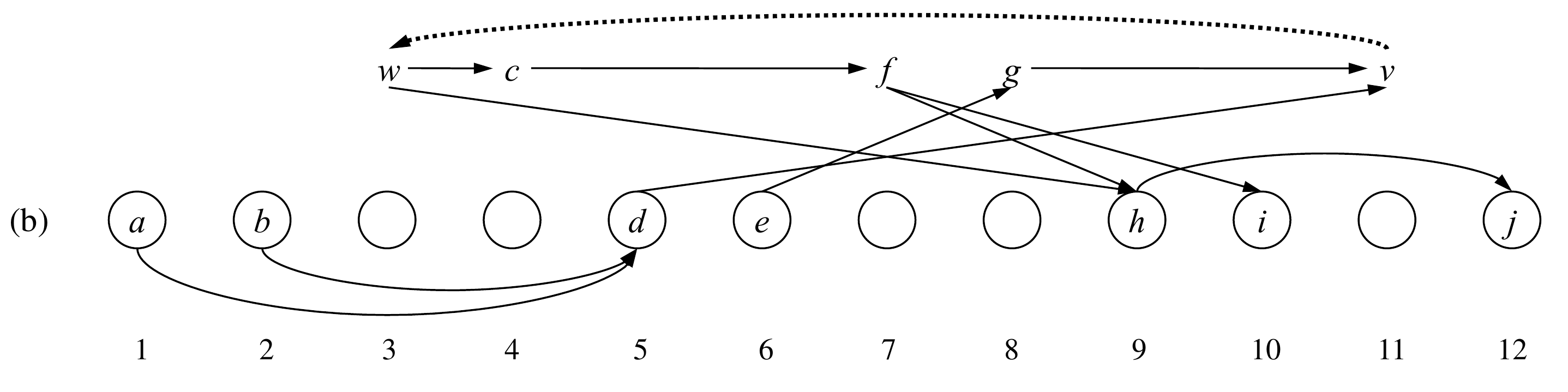}\\
\vspace{5pt}
\includegraphics[scale=0.45]{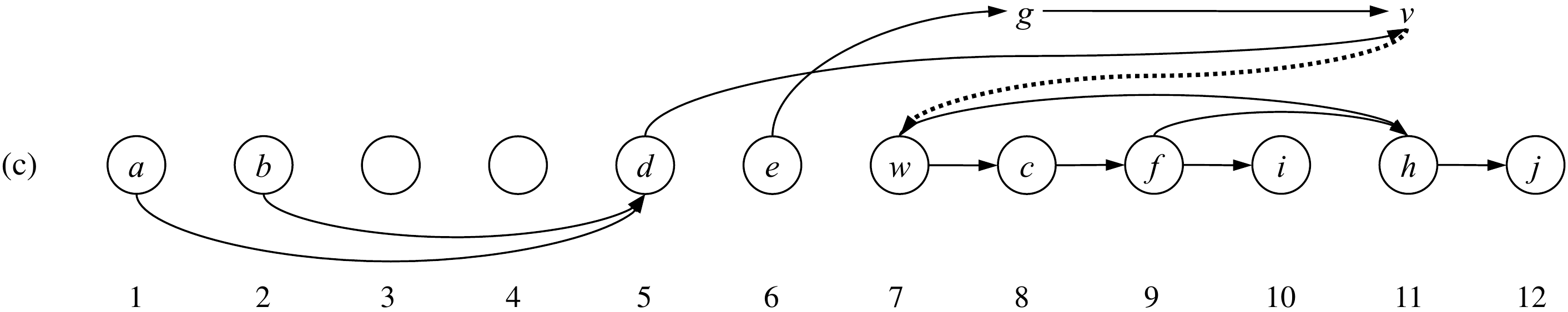}\\
\vspace{5pt}
\includegraphics[scale=0.45]{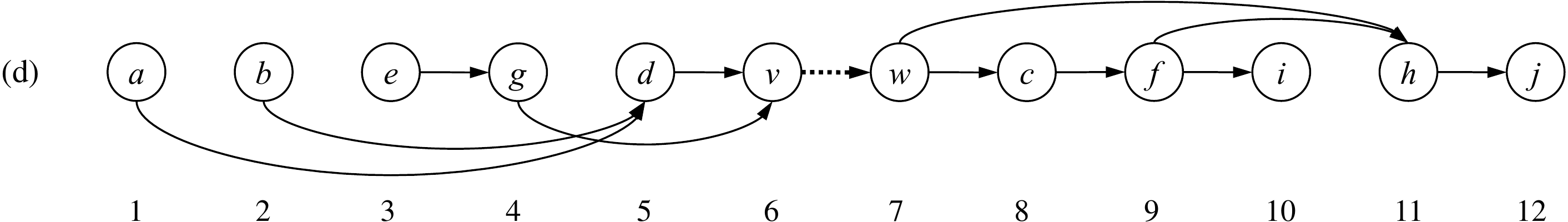}\\
\caption{
Topological search and reordering of the graph in Figure~\ref{fig:lim-search}.
(a) Initially positions 3 and 11, of $w$ and $v$, respectively, become empty,
$F = [w]$, $B = [v]$. The search adds $c$ to $F$, $g$ to $B$, $f$ to $F$, and
stops with $i = j = 7$. (b) Forward reordering begins from position 7.  Vertex
$w$ drops into position 7, $c$ drops into position 8.  Vertex $h$ in position 9
has an arc from $f$, still in $F$: $h$ is added to $F$, $f$ drops into position
9. Vertex $i$ in position 10 has no arc from any vertex still in $F$.  Vertex
$h$ drops into position 11. (c) Backward reordering begins from position 6. 
Vertex $e$ has an arc to $g$; $e$ is added to $B$ and replaced by $v$. Vertices
$g$ and $e$ drop into positions 4 and 3, respectively.  (d) Final order.  Forward
and backward reordering are independent and can be done in either order or
concurrently.
}
\label{fig:top-search}
\end{figure}

\begin{theorem}
\label{thm:top-corr}
Topological search is correct.

\end{theorem}

\begin{proof}
Let $(v, w)$ be a new arc such that $v > w$.  The search maintains the invariant
that every vertex in $F$ is reachable from $w$ and $v$ is reachable from every vertex in
$B$.  Thus if there is an arc $(u, z)$ with $u$ in $F$ and $z$ in $B$, there is a cycle.
Suppose the addition of $(v, w)$ creates a cycle.  The cycle consists of $(v,
w)$ and a path $P$ from $w$ to $v$ of vertices in increasing order.  Let $u$ be the
largest vertex on $P$ that is in $F$ at the end of the search.  Since $u \ne v$, there is an arc $(u, z)$
on $P$.  Vertex $z$ must be in $B$, or the search would not have stopped.  We conclude that
the algorithm reports a cycle if and only if the addition of $(v, w)$ creates one.

Suppose the addition of $(v, w)$ does not create a cycle.  When the search stops, the
number of positions $g \ge i$ such that $\vertex(g) = \nlv$ is $|F|$.  The
forward reordering loop maintains this invariant as it updates $F$.  It also maintains the invariant that
once position $i$ is processed, every position from $k$ to $i$, inclusive, is non-null.
Thus if $i = n + 1$, $F$ must be empty, and the loop terminates.  Symmetrically, the
backward reordering loop terminates before $j$ can become 0.  Thus all vertices in $F \cup
B$ at the end of the search are successfully reordered; some other vertices may also
be reordered.  Let $\overline{F}$ and $\overline{B}$ be the sets of vertices added to $F$ and to $B$,
respectively, during the search and reordering.  Vertices in $\overline{F}$ move to higher
positions, vertices in $\overline{B}$ move to lower positions, and no other vertices move.

All vertices in $\overline{F}$ are reachable from $w$, and $v$ is reachable from all
vertices in $\overline{B}$.  We show by case analysis that after the reordering every
arc $(x, y)$ has $x < y$.  There are five cases, of which two pairs are symmetric.
Suppose $x$ and $y$ are both in $\overline{F} \cup \overline{B}$.  Since there is no
cycle, it cannot be the case that $x$ is in $\overline{F}$ and $y$ is in
$\overline{B}$.  The reordering moves all vertices in $\overline{F}$ after all
vertices in $\overline{B}$ without changing the order of vertices in $\overline{F}$
and without changing the order of vertices in $\overline{B}$.  It follows that $x <
y$ after the reordering.  This includes the case $(x, y) = (v, w)$, since $w$ is in
$\overline{F}$ and $v$ is in $\overline{B}$.  Suppose $y$ is in $\overline{F}$
and $x$ is not in $\overline{F} \cup \overline{B}$. 
The reordering does not move $x$ and moves $y$ higher in the order, so $x < y$ after the reordering.  The case of $x$
in $\overline{B}$ and $y$ not in $\overline{F} \cup \overline{B}$ is symmetric.
Suppose $x$ is in $\overline{F}$ and $y$ is not in $\overline{F} \cup \overline{B}$.
Since $x < y$ before the reordering, the first loop of the reordering must reinsert
$x$ before it reaches the position of $y$; otherwise $y$ would be in $\overline{F}$.
Thus $x < y$ after the reordering.  The case $y$ in $\overline{B}$ and $x$ not in
$\overline{F} \cup \overline{B}$ is symmetric.  \qed

\end{proof}

To bound the running time of topological search, we extend the concept of
relatedness to vertex pairs. We say two vertices are {\em related} if they are
on a common path.  Relatedness is symmetric; order on the path does not matter.

\begin{lemma}
\label{lem:top-cycle-time}
Over $m$ arc additions, topological search spends $\bigO(n^2)$ time testing for cycles.

\end{lemma}

\begin{proof}
Suppose addition of an arc $(v, w)$ triggers a search.  Let $F$ and $B$ be the values
of the corresponding variables at the end of the search.  The test for cycles takes
$\bigO(|F||B|)$ time.  If this is the last arc addition, the test takes $\bigO(n^2)$ time.  Each
earlier addition does not create a cycle; for such an addition, each pair $x$ in $F$
and $y$ in $B$ is related after the addition but not before: before the reordering $x <
y$, so if $x$ and $y$ were related there would be a path from $x$ to $y$, and the addition of
$(v, w)$ would create a cycle, consisting of a path from $w$ to $x$, the path
from $x$ to $y$, a path from $y$ to $v$, and arc $(v,w)$. Since there are at
most ${n \choose 2}$ related vertex pairs, the time for all cycle tests other than the last is $\bigO(n^2)$.  \qed

\end{proof}

For each move of a vertex during reordering, we define the {\em distance} of the move
to be the absolute value of the difference between the positions of the vertex in the
old and new orders.

\begin{lemma}
\label{lem:sum-distances}
Over all arc additions, except the last one if it creates a cycle, the time spent by
topological search doing search and reordering is at most a constant times the
sum of the distances of all the vertex moves.

\end{lemma}

\begin{proof}
Consider an arc addition that triggers a search and reordering.  Consider a vertex $q$
that is moved to a higher position; that is, it is added to $F$ during either the
search or the reordering and eventually placed in a new position during the
reordering.  Let $i_1$ be its position before the reordering and $i_2$ its
position after the reordering.  When $q$ is added to $F$, $i = i_1$; when $q$ is removed
from $F$, $i = i_2$.  For each value of $i$ greater than $i_1$ and no greater
than $i_2$, there may be a test for the existence of an arc $(q, \vertex(i))$: such a test
can occur during forward search or forward reordering but not both.  The number of
such tests is thus at most $i_2 - i_1$, which is the distance $q$ moves.  A
symmetric argument applies to a vertex moved to a lower position.  Every test for an
arc is covered by one of these two cases.  Thus the number of arc tests is at most
the sum of the distances of vertex moves.  The total time spent in search and
reordering is $\bigO(1)$ per increment of $i$, per decrement of $j$, and per arc test.  For
each increment of $i$ or decrement of $j$ there is either an arc test or an insertion of
a vertex into its final position.  The number of such insertions is at most one per
vertex moved.  The lemma follows.  \qed

\end{proof}

It remains to analyze the sum of the distances of the vertex moves. To simplify the
analysis, we decompose the moves into pairwise swaps of vertices.  Consider sorting a
permutation of $1$ through $n$ by doing a sequence of pairwise swaps of out-of-order
elements.  The {\em distance} of a swap is twice the absolute value of the difference
between the positions of the swapped elements; the factor of two accounts for the two
elements that move.  The sequence of swaps is {\em proper} if, once a pair is
swapped, no later swap reverses its order.

Consider the behavior of topological search over a sequence of arc additions,
excluding the last one if it creates a cycle. Identify the vertices with their final
positions.  Then the topological order is a permutation, and the final permutation is
sorted.

\begin{lemma}
\label{lem:proper-swaps}
There is a proper sequence of vertex swaps whose total distance equals the sum of the
distances of all the reordering moves.

\end{lemma}

\begin{proof}
Consider an arc addition that triggers a search and reordering.  As in the proof of
Theorem~\ref{thm:top-corr}, let $\overline{F}$ and $\overline{B}$ be the sets of
vertices added to $F$ and to $B$, respectively, during the search and reordering.
Consider the positions of the vertices in $\overline{F} \cup \overline{B}$ before and
after the reordering.  After the reordering, these positions from lowest to highest
are occupied by the vertices in $\overline{B}$ in their original order, followed by
the vertices in $\overline{F}$ in their original order.  We describe a sequence of
swaps that moves the vertices in $\overline{F} \cup \overline{B}$ from their positions
before the reordering to their positions after the reordering.  Given the
outcome of the swaps so far, the next swap is of any two vertices $x$ in $\overline{F}$ and $y$
in $\overline{B}$ such that $x$ is in a smaller position than $y$ and no vertex in
$\overline{F} \cup \overline{B}$ is in a position between that of $x$ and that of $y$.
The swap of $x$ and $y$ moves $x$ higher, moves $y$ lower, and preserves the order of
the vertices in $\overline{F}$ as well as the order of the vertices in
$\overline{B}$.  If no swap is possible, all vertices in $\overline{F}$ must follow
all vertices in $\overline{B}$, and since swaps preserve the order within
$\overline{F}$ and within $\overline{B}$ the vertices are now in their positions
after the reordering.  Only a finite number of swaps can occur, since each vertex can
only move a finite distance (higher for a vertex in $\overline{F}$, lower for a
vertex in $\overline{B}$).  The total distance of the moves of the vertices in
$\overline{F}$ is exactly half the distance of the swaps, as is the total distance of
the moves of the vertices in $\overline{B}$.  Any particular pair of vertices is
swapped at most once.  Repeat this construction for each arc addition.  If an arc
addition causes a swap of $x$ and $y$, with $x$ moving higher and $y$ moving lower,
then the arc addition creates a path from $y$ to $x$, and no later arc addition can
cause a swap of $x$ and $y$.  Thus the swap sequence is proper.  \qed

\end{proof}

The following lemma was proved by Ajwani et al. \cite{Ajwani2006} as part of the
analysis of their $\bigO(n^{11/4})$-time algorithm.  Their proof uses a linear
program. We give a combinatorial argument.

\begin{lemma}
\label{lem:swap-distance}
{\em \cite{Ajwani2006}} Given an initial permutation of $1$ through $n$, any proper
sequence of swaps has total distance $\bigO(n^{5/2})$.

\end{lemma}

\begin{proof}
If $\RPi$ is a permutation of 1 to $n$, we denote by $\RPi(i)$ the $i^{\text{th}}$
element of $\RPi$.  We define the {\em potential} of $\RPi$ to be $\sum_{i<j}
(\RPi(i) - \RPi(j))$.  The potential is always between $-n^3$ and $n^3$.  We
compute the change in potential caused by a swap in a proper swap sequence. 
Let $\RPi$ be the permutation before the swap, and let $i < j$ be the positions
in $\RPi$ of the pair of elements ($\RPi(i)$ and $\RPi(j)$) that are swapped. 
The distance $d$ of the swap is $2(j - i)$.  Since the swap sequence is proper,
$\RPi(i) > \RPi(j)$.  Swapping $\RPi(i)$ and $\RPi(j)$ reduces the contribution
to the potential of the pair $i,j$ by $2(\RPi(i) - \RPi(j))$.  The swap also
changes the contributions to the potential of pairs other than $i,j$,
specifically those pairs exactly one of whose elements is $i$ or $j$.  We
consider three cases for the other element of the pair, say $k$.  If $k < i$,
the swap increases the contribution of $k,i$ by $\RPi(i) - \RPi(j)$ and
decreases the contribution of $k,j$ by $\RPi(i) - \RPi(j)$, for a net change of
zero.  Similarly, if $j < k$, the swap decreases the contribution of $i,k$ by
$\RPi(i) - \RPi(j)$ and increases the contribution of $j,k$ by $\RPi(i)
- \RPi(j)$, for a net change of zero. More interesting is what happens if $i < k
< j$.  In this case the swap decreases the contribution of both $i,k$ and $k,j$
by $\RPi(i) - \RPi(j)$.  There are $j - i - 1$ such values of $k$.  Summing
over all pairs, we find that the swap decreases the potential of the
permutation by $2(\RPi(i) - \RPi(j))( 1 + j - i - 1) = d(\RPi(i) - \RPi(j))$.

Call a swap of $\RPi(i)$ and $\RPi(j)$ {\em small} if $\RPi(i) - \RPi(j) < \sqrt{n}$
and {\em big} otherwise.  Because the swap sequence is proper, a given pair
can be swapped at most once.  Thus there are $\bigO(n^{3/2})$ small swaps.  Each has
distance at most $2(n - 1)$, so the sum of the distances of all small swaps is
$\bigO(n^{5/2})$.  A big swap of distance $d$ reduces the potential by at least
$d\sqrt{n}$.  Since the sum of the potential decreases over all swaps is
$\bigO(n^3)$, the sum of the distances of all big swaps is $\bigO(n^{5/2})$.  \qed

\end{proof}

The proof of Lemma \ref{lem:swap-distance} does not require that the swap sequence be
proper; it suffices that every swap is of an out-of-order pair and no pair of
elements is swapped more than once.  The lemma may even be true if all swaps
are of out-of-order pairs with some pairs swapped repeatedly, but our proof
fails in this case, because our bound on the distance of the small swaps
requires that there be $\bigO(n^{3/2})$ of them.

\begin{theorem}
\label{thm:top-search-time}
Over $m$ arc additions, topological search spends $\bigO(n^{5/2})$ time.

\end{theorem}

\begin{proof}
Topological search spends $\bigO(n^2)$ time on the last arc addition.  By
Lemmas~\ref{lem:top-cycle-time}-\ref{lem:swap-distance}, it spends
$\bigO(n^{5/2})$ on all the rest. \qed

\end{proof}

The bound of Theorem~\ref{thm:top-search-time} may be far from tight.  In the
remainder of this section we discuss lower bounds on the running time of topological
search, and we speculate on improving the upper bound.

Katriel~\cite{Katriel2004} showed that any topological sorting algorithm that is
local (as defined in Section~\ref{sec:lim-search}: the algorithm reorders
only affected vertices) must do $\Omga(n^2)$ vertex renumberings on a sequence of arc
additions that form a path.  This bound is $\Omga(n)$ amortized per arc on a graph of
$\bigO(n)$ arcs.  She also proved that the topological sorting algorithm of
Pearce and Kelly~\cite{Pearce2006} does $\bigO(n^2)$ vertex renumberings.  Since topological
search is a local topological sorting algorithm, her lower bound applies to this
algorithm.  Her lower bound on vertex reorderings is tight for topological search,
since a proper sequence of swaps contains at most ${n \choose 2}$ swaps, and
each pair of reorderings corresponds to at least one swap.

To get a bigger lower bound, we must bound the total distance of vertex moves, not
their number.  Ajwani~\cite{AjwaniThesis} gave a result for a related problem that
implies the following: on a sequence of arc additions that form a path, topological
search can take $\Omga(n^2\log n)$ time.  We proved this result independently in
our conference paper~\cite{Haeupler2008b}; our proof uses the same construction
as Ajwani's proof. This bound is $\Omga(n\log n)$ amortized per arc on a graph
of $\bigO(n)$ arcs.

We do not know if Ajwani's bound is tight for graphs with $\bigO(n)$ arcs, but it is
not tight for denser graphs.  There is an interesting connection between the running
time of topological search and the notorious $k$-levels problem of computational
geometry.  Uri Zwick~(private communication, 2009) pointed this out to us. The
$k$-levels problem is the following: Consider the intersections of $n$ lines in
the plane in general position: each intersection is of only two lines, and the
intersections have distinct $x$-coordinates.  An intersection is a {\em
$k$-intersection} if there are exactly $k$ lines below it (and $n - k - 2$
lines above it).  What is the maximum number of $k$-intersections as a function
of $n$ and $k$?  For our purposes it suffices to consider $n$ even and $k = n/2
- 1$.  We call an intersection with $n/2 - 1$ lines below it a {\em halving
intersection}. The current best upper and lower bounds on the maximum number of
halving intersections are $\bigO(n^{4/3})$~\cite{Dey1998} and
$\Omga(n2^{\sqrt{2\lg n}}/\sqrt{\lg n})$ (~\cite{Nivasch2008}; see also~\cite{Toth2001}).

The relationship between the $k$-levels problem and our problem does not require
that the lines be straight; it only requires that each pair intersect only once.
Thus instead of a set of lines we consider a set of {\em pseudolines},
arbitrary continuous functions from the real numbers to the real numbers, each
pair of which intersect at most once.  Such a set is in {\em general position}
if no point is common to three or more pseudolines, no two intersections of
pseudolines have the same $x$-coordinate, and each intersection is a crossing
intersection: if pseudolines $P$ and $Q$ intersect and $P$ is above $Q$ to the
left of the intersection, then $Q$ is above $P$ to the right of the
intersection.  The best bounds on the number of halving intersections of $2n$
pseudolines in general position are $\bigO(n^{4/3})$ (~\cite{TamakiT2003}; see
also~\cite{SharirS2003}) (the same as for lines) and $\Omga(n2^{\sqrt{2\lg
n}})$~\cite{Zwick2005}. The latter gives a lower bound of $\Omga(n^2
2^{\sqrt{2\lg n}})$ on the worst-case running time of topological search, as we
now show.

\begin{theorem}
\label{thm:halving-lb}
Let $n$ be even.  On a graph of $3n/2$ vertices, topological search can spend
$\Omga(n)$ time per arc addition for at least $H(n)$ arc additions, where $H(n)$
is the maximum number of halving intersections of $n$ pseudolines in the plane
in general position.

\end{theorem}

\begin{proof}
Given a set of $n$ pseudolines with $H(n)$ halving intersections, we construct a
sequence of $H(n)$ arc additions on a graph of $3n/2$ vertices on which
topological search spends $\Omga(n)$ time on each arc addition.  Given such a
set of pseudolines, choose a value $x_0$ of the $x$-coordinate sufficiently
small that all the halving intersections have $x$-coordinates larger than $x_0$.  Number
the pseudolines from 1 to $n$ from highest to lowest $y$-coordinate at $x_0$, so
that the pseudoline with the highest $y$-coordinate gets number 1 and the one with
the lowest gets number $n$.  Construct a graph with $3n/2$ vertices and an initial
(arbitrary) topological order.  Number the first $n/2$ vertices in order from
$n$ down to $n/2 + 1$, and number the last $n/2$ vertices in order from $n/2$
down to 1, so that the first vertex gets number $n$, the $(n/2)^{\text{th}}$
gets number $n/2 + 1$, the middle $n/2$ get no number, the $(n + 1)^{\text{st}}$
gets number $n/2$, and the last gets number 1.  These numbers are permanent and
are a function only of the initial order.  Identify vertices by their number. 
Process the halving intersections in order by $x$-coordinate.  If the
$k^{\text{th}}$ halving intersection is of pseudolines $i$ and $j$ with $i < j$,
add an arc $(i, j)$ to the graph.  To the left of the intersection, pseudoline
$i$ is above pseudoline $j$; to the right of the intersection, pseudoline $j$ is
above pseudoline $i$. Figure~\ref{fig:halving-lb} illustrates this construction. 

\begin{figure}
\centering
\includegraphics[scale=0.45]{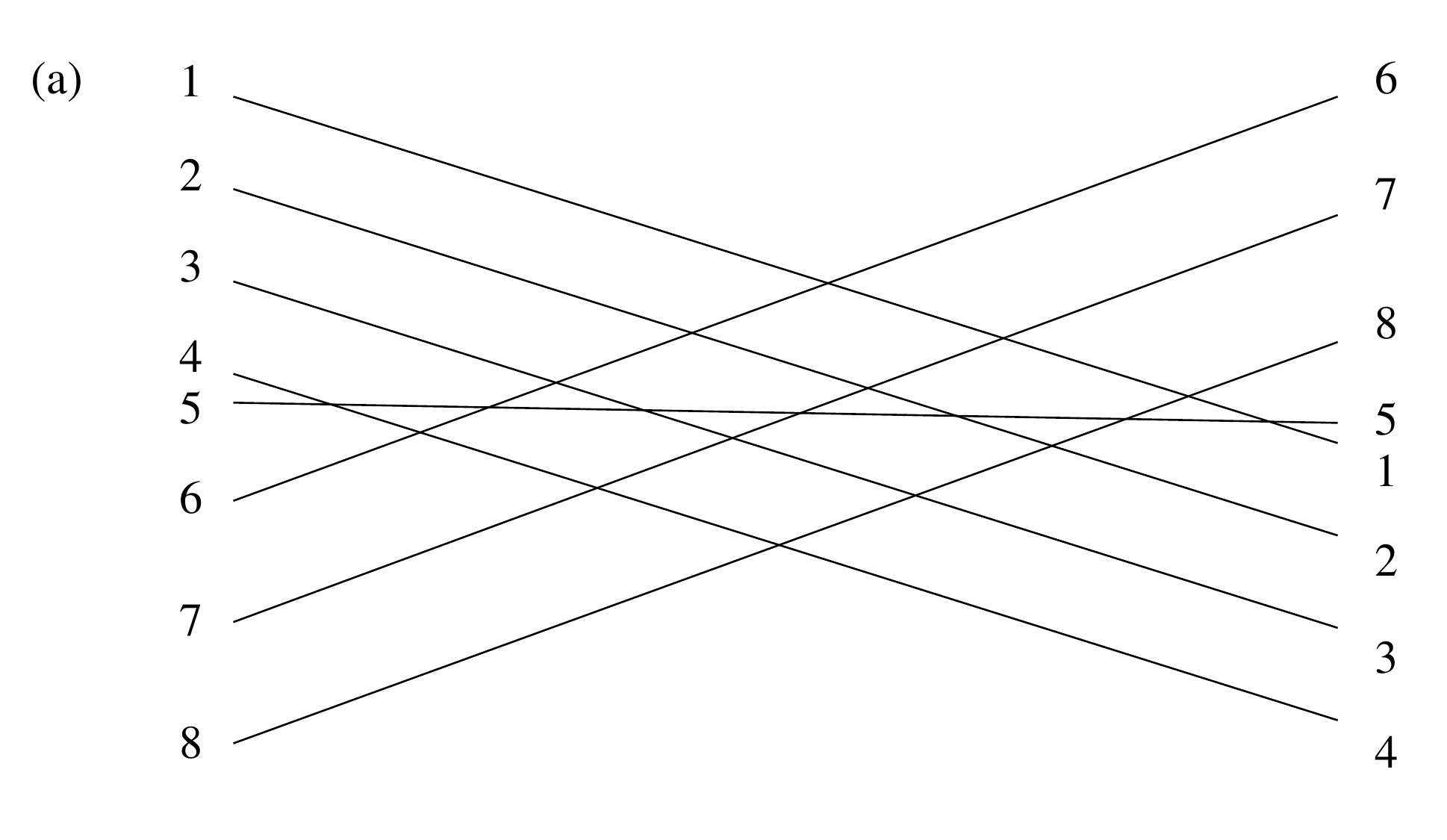}\\
\vspace{5pt}
\includegraphics[scale=0.45]{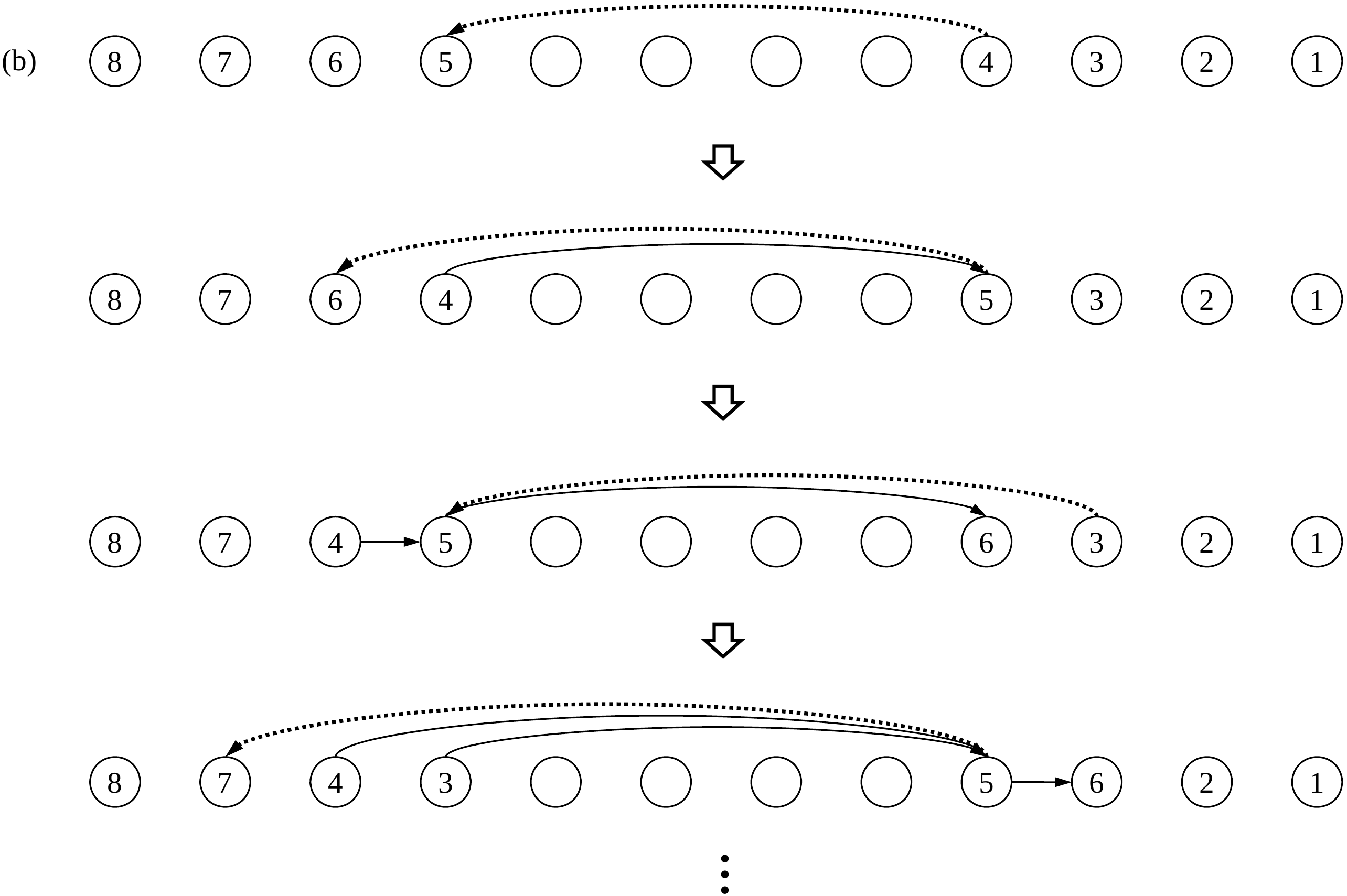}
\caption{(a) A set of $n = 8$ pseudolines with $H(n) = 7$ halving intersections.
Although the pseudolines are straight in this example, in general they need not
be. (b) The corresponding sequence of arc additions on a graph of $3n/2 = 12$
vertices on which topological search takes $\Omega(nH(n))$ time.  The arc
additions correspond to the halving intersections processed in increasing order
by $x$-coordinate; only the first four arc additions are shown.}
\label{fig:halving-lb}
\end{figure}

Since each arc $(i, j)$ has $i < j$, the graph remains acyclic.  Since two
pseudolines have only one intersection, a given arc is added only once. 
Consider running topological search on this set of arc additions.  We claim that each arc
addition moves exactly one vertex from the last third of the topological order
to the first third and vice-versa; the vertices in the middle third are never
reordered.  Each such arc addition takes $\Omga(n)$ time, giving the theorem.

To verify the claim, we prove the following invariant on the number of arc
additions: the set of vertices in the first or last third, respectively, of the
topological order have the same numbers as the bottom or top half of the
pseudolines, respectively.  In particular, a halving intersection of two
pseudolines $i, j$ with $i < j$ corresponds to a swap of vertices $i$ in the top
third and $j$ in the bottom third, giving the claim.

Intersections that are not halving intersections preserve the invariant. 
Suppose the invariant is true just to the left of a halving intersection of
pseudolines $i$ and $j$ with $i < j$.  Just to the left of the intersection,
pseudolines $i$ and $j$ are the $n/2$ and $n/2 + 1$ highest pseudolines,
respectively. By the induction hypothesis, just before the addition of $(i, j)$ vertex $i$ is in the
last third of the topological order and vertex $j$ is in the first third. 
Suppose that just before the addition of $(i, j)$ there is an arc $(j, k)$ with
$k$ in the first third.  Then $j < k$, but pseudoline $k$ is in the bottom half
and hence must be below pseudoline $j$. This is impossible, since the existence
of the arc $(j, k)$ implies that pseudoline $k$ crossed above pseudoline $j$ to
the left of the intersection of $i$ and $j$.  Thus there can be no such arc $(j, k)$. 
Symmetrically, there can be no arc $(k, i)$ with $k$ in the last third.  It
follows that the topological search triggered by the addition of $(i, j)$ will
compute a set $F$ all of whose vertices except $j$ are in the last third and a
set $B$ all of whose vertices except $i$ are in the first third.  The
subsequent reordering will move $i$ to the first third, move $j$ to the last
third, and possibly reorder other vertices within the first and last thirds. 
Thus the invariant remains true after the addition of $(i, j)$. By induction
the invariant holds, giving the claim, and the theorem.  \qed

\end{proof}

\begin{corollary}
There is a constant $c > 0$ such that, for all $n$, there is a
sequence of arc additions on which topological search takes $cn^2
2^{\sqrt{2\lg n}}$ time.

\end{corollary}

Unfortunately the reduction in the proof of Theorem~\ref{thm:halving-lb} goes
only one way.  We have been unable to construct a reduction in the other
direction, nor are we able to derive a better upper bound for topological
search via the methods used to derive upper bounds on the number of halving
intersections.

\section{Strong Components} \label{sec:strong}

All the known topological ordering algorithms can be extended to the maintenance of
strong components with at most a constant-factor increase in running time.  Pearce
\cite{Pearce2005} and Pearce and Kelly \cite{Pearce2003b} sketch how to extend their
algorithm and that of Marchetti-Spaccamela et al.~\cite{Marchetti1996} to strong
component maintenance. Here we describe how to extend soft-threshold search and
topological search.  The techniques differ slightly for the two algorithms, since one
algorithm is designed for the sparse case and the other for the dense case.

We formulate the problem as follows: Maintain the partition of the vertices defined
by the strong components.  For each strong component, maintain a {\em canonical
vertex}.  The canonical vertex represents the component; the algorithm is free to
choose any vertex in the component to be the canonical vertex.  Support the query
$\find(v)$, which returns the canonical vertex of the component containing vertex $v$.
Maintain a list of the canonical vertices in a topological order of the corresponding
components.

To represent the vertex partition, we use a disjoint set data
structure~\cite{Tarjan1975,Tarjan1984}.  This structure begins with the partition
consisting of singletons and supports find queries and the operation $\unite(x, y)$,
which, given canonical vertices $x$ and $y$, forms the union of the sets containing
$x$ and $y$ and makes $x$ the canonical vertex of the new set.  If the sets are
represented by trees, the finds are done using path compression, and the unites are
done using union by rank, the amortized time per find is $\bigO(1)$ if the total time
charged for the unites is $\bigO(n\log n)$~\cite{Tarjan1984}.  (In fact, the time
charged to the unites can be made much smaller, but this weak bound suffices for us.)
Since searching and reordering take much more than $\Omga(n\log n)$ time, we can
treat the set operations as taking $\bigO(1)$ amortized time each.

To maintain strong components using soft-threshold search, we represent the graph by
storing, for each canonical vertex $x$, a list of arcs out of its component, namely
those arcs $(y, z)$ with $\find(y) = x$, and a list of arcs into its component, namely
those arcs $(y, z)$ with $find(z) = x$.  This represents the graph of strong components,
except that there may be multiple arcs between the same pair of strong components,
and there may be loops, arcs whose ends are in the same component.  When doing a
search, we delete loops instead of traversing them.  When the addition of an arc $(v,
w)$ combines several components into one, we form the incoming list and the outgoing
list of the new component by combining the incoming lists and outgoing lists,
respectively, of the old components.  This takes $\bigO(1)$ time per old component, if
the incoming and outgoing lists are circular.  Deletion of a loop takes $\bigO(1)$ time
if the arc lists are doubly linked.

Henceforth we identify each strong component with its canonical vertex, and we
abbreviate $\find(x)$ by $f(x)$.  If a new arc $(v, w)$ has $f(v) > f(w)$, do a
soft-threshold search forward from $f(w)$ and backward from $f(v)$.  During the
search, do not stop when a forward arc traversal reaches a component in $B$ or 
when a backward arc traversal reaches a component in $F$.  Instead, allow
components to be in both $F$ and $B$.  Once the search stops, form the new
component, if any. Then reorder the canonical vertices and delete from the
order those that are no longer canonical.  Here are the details.  When a new
arc $(v, w)$ has $f(v) > f(w)$, do the search by calling
\softthresholdsearch{$f(v)$,$f(w)$}, where \softthresholdsearch is defined
as in Section~\ref{sec:soft-search} but with the macro \searchstep
redefined as in Figure~\ref{alg:search-step-strong}.  The new version of
\searchstep is just like the old one except that it visits canonical vertices
instead of all vertices, it uses circular instead of linear arc lists, and it
does not do cycle detection: \softthresholdsearch terminates only when $F_A$
or $B_A$ is empty, and it always returns null.

\begin{figure}
\setlength{\algomargin}{8.5em}
\begin{procedure}[H]
{\bf macro} \searchstep{{\rm vertex} $u$, {\rm vertex} $z$}\;
\Indp
    $(q, g) = \outl(u)$; $(h, r) = \inl(z)$\;
    $\outl(u) = \noutl((q, g))$; $\inl(z) = \ninl((h, r))$\; 
    $x = f(g)$; \lIf {$\outl(u) = \foutl(u)$} {$F_A = F_A \sm \{u\}$}\; 
    $y = f(h)$; \lIf {$\inl(z) = \finl(z)$} {$B_A = B_A \sm \{z\}$}\; 
    \lIf {$u = x$} {delete $(q, g)$}; \lIf {$y = z$} {delete $(h, r)$}\; 
    \If {$x \not\in F$} {
    	$F = F \cup \{x\}$; $\outl(x) = \foutl(x)$\;
    	\lIf {$\outl(x) \ne \nlv$} {$F_A = F_A \cup \{x\}$}\;
    }
    \If {$y \not\in B$} {
    	$B = B \cup \{y\}$; $\inl(y) = \finl(y)$\; 
    	\lIf {$\inl(y) \ne \nlv$} {$B_A = B_A \cup \{y\}$}\;
    }
\end{procedure}
\caption{Redefinition of \searchstep to find strong components using
soft-threshold search.}
\label{alg:search-step-strong}
\end{figure}

Once the search finishes, let $t = \min(\{f(v)\} \cup \{x \in F| \outl(x) \ne
\nlv\})$.  Compute the sets $F_<$ and $B_>$. Find the new component, if any, by
running a static linear-time strong components algorithm on the subgraph of the
graph of strong components whose vertex set is $X = F_< \cup \{t\} \cup B_>$
and whose arc set is $Y = \{(f(u), f(x))|(u, x) \text{ is an arc with } f(u)
\text{ in } \linebreak[0] F_< \text{ and } f(u) \ne f(x)\} \cup \{(f(y),
f(z))|(y, z) \text{ is an arc with } f(z) \in B_> \text{ and } f(y) \ne f(z)\}$.  If a new
component is found, combine the old components it contains into a new component
with canonical vertex $v$.

Reorder the list of vertices in topological order by moving the vertices in $X -
\{t\}$ as in Section~\ref{sec:2way-search}.  Then delete from the list all
vertices that are no longer canonical, namely the canonical vertices other than $f(v)$ of the
old components contained in the new component.  

{\em Remark}: Since the addition of $(v, w)$ can only form a single new component,
running a strong components algorithm to find this component is overkill.  A simpler
alternative is to unmark all vertices in $X$ and then run a forward depth-first
search from $f(w)$, traversing arcs in $Y$.  During the search, mark vertices as
follows: Mark $f(v)$ if it is reached. When retreating along an arc $(f(u), f(x))$,
mark $f(u)$ if $f(x)$ is marked.  At the end of the search, the marked vertices are
the canonical vertices contained in the new component.

\begin{theorem}
\label{thm:strong-soft-corr}
Maintaining strong components via soft-threshold search is correct.

\end{theorem}

\begin{proof}
By induction on the number of arc additions.  Consider the graph of strong components
just before an arc $(v, w)$ is added.  This addition forms a new component if and
only if $f(v) > f(w)$ and there is a path from $f(w)$ to $f(v)$.  Furthermore the old
components contained in the new component are exactly the components on paths
from $f(w)$ to $f(v)$. The components on such a path are in increasing order,
so the path consists of a sequence of one or more components in $F_<$, possibly
$t$, and a sequence of one or more components in $B_>$.  Each arc on such a
path is in $Y$.  It follows that the algorithm correctly forms the new
component.  If there is no new component, the reordering is exactly the same as
in Section~\ref{sec:2way-search}, so it correctly restores topological order. 
Suppose there is a new component.  Then certain old components are combined
into one, and their canonical vertices other than $f(v)$ are deleted from the
list of canonical vertices in topological order.  We must show that the new
order is topological.  The argument in the proof of Theorem~\ref{thm:2way-corr}
applies, except that there are some new possibilities. Consider an arc $(x, y)$
other than $(v, w)$.  One of the cases in the proof of
Theorem~\ref{thm:2way-corr} applies unless at least one of $x$ and $y$ is in
the new component.  If both are in the new component, then $(x, y)$ becomes a
loop.  Suppose just one, say $y$, is in the new component.  Then $f(x)$ cannot
be forward, or it would be in the new component.  Either $f(x)$ is in $B_>$ or
$f(x)$ is not in $X$; in either case, $f(x)$ precedes $f(v)$ after the
reordering.  The argument is symmetric if $x$ but not $y$ is in the new
component.  \qed

\end{proof}

To bound the running time of the strong components algorithm we need to extend
Lemma~\ref{lem:2way-search-rel} and Theorem~\ref{thm:2way-search-arcs}.

\begin{lemma}
\label{lem:strong-search-rel}
Suppose the addition of $(v, w)$ triggers a search.  Let $(u, x)$ and $(y, z)$,
respectively, be arcs traversed forward and backward during the search, not
necessarily during the same search step, such that $f(u) < f(z)$.  Then either $(u,
x)$ and $(y, z)$ are unrelated before the addition of $(v, w)$ but related afterward,
or they are related before the addition and the addition makes them into loops.

\end{lemma}

\begin{proof}
After $(v, w)$ is added, there is a path containing both of them, so they are related
after the addition.  If they were related before the addition, then there must be a
path containing $(u, x)$ followed by $(y, z)$.  After the addition there is a path
from $z$ to $u$, so $u$, $x$, $y$, and $z$ are in the new component, and both $(u,
x)$ and $(y, z)$ become loops. \qed

\end{proof}

\begin{theorem}
\label{thm:strong-arcs}
Over $m$ arc additions, the strong components algorithm does $\bigO(m^{3/2})$ arc
traversals.

\end{theorem}

\begin{proof}
Divide the arc traversals during a search into those of arcs that become loops as a
result of the arc addition that triggered the search, and those that do not.  Over
all searches, there are at most $2m$ traversals of arcs that become loops: each
such arc can be traversed both forward and backward.  By
Lemma~\ref{lem:strong-search-rel} and the proof of
Theorem~\ref{thm:2way-search-arcs}, there are at most $4m^{3/2}$ traversals of
arcs that do not become loops.  \qed

\end{proof}

\begin{theorem}
Maintaining strong components via soft-threshold search takes \linebreak
$\bigO(m^{3/2})$ time over $m$ arc additions, worst-case if $s$ is always a median or approximate
median of the set of choices, expected if $s$ is always chosen uniformly at
random.

\end{theorem}

\begin{proof}
Consider the addition of an arc $(v, w)$ such that $f(v) > f(w)$.  Each search step
either traverses two arcs or deletes one or two loops. An arc can only become a
loop once and be deleted once, so the extra time for such events is $\bigO(m)$
over all arc additions.  The arcs in $Y$ were traversed by the search, so the
time to form the new component and to reorder the vertices is $\bigO(1)$ per
arc traversal.  The theorem follows from Theorem~\ref{thm:strong-arcs} and the
proofs of Theorems~\ref{thm:soft-search-time} and~\ref{thm:soft-search-rtime}. 
\qed

\end{proof}

To maintain strong components via topological search, we represent the graph of
strong components by an adjacency matrix $A$ with one row and one column per
canonical vertex.  If $x$ and $y$ are canonical vertices, $A(x, y) = 1$ if $x \ne y$
and there is an arc $(q, r)$ with $f(q) = x$ and $f(r) = y$; otherwise, $A(x, y) =
0$.  We represent the topological order of components by an explicit numbering of the
canonical vertices using consecutive integers starting from one.  We also store the
inverse of the numbering.  If $x$ is a canonical vertex, $\position(x)$ is
its number; if $i$ is a vertex number, $\vertex(i)$ is the canonical vertex with
number $i$. Note that the matrix $A$ is indexed by vertex, {\em not} by vertex number; the
numbers change too often to allow indexing by number.

To maintain strong components via topological search, initialize all entries of $A$
to zero.  Add a new arc $(v, w)$ by setting $A(f(v), f(w)) = 1$.  If $f(v) > f(w)$,
search forward from $f(w)$ and backward from $f(v)$ by executing
\topologicalsearch{$f(v)$, $f(w)$} where \topologicalsearch is defined as in
Section~\ref{sec:top-search}.  Let $k$ be the common value of $i$ and $j$ when
the search stops.  After the search, find the vertex set of the new component,
if any, by running a linear-time static strong components algorithm on the graph
whose vertex set is $X = F \cup B$ and whose arc set is $Y = \{(x, y)|x
\text{ and } y \text{ are in } F \cup B \text{ and } A(x, y) = 1\}$.  Whether or
not there is a new component, reorder the old canonical vertices exactly as in
Section~\ref{sec:top-search}.  Finally, if there is a new component, do the
following: form its vertex set by combining the vertex sets of the old
components contained in it.  Let the canonical vertex of the new component be
$\vertex(k)$. Form a row and column of $A$ representing the arcs out of and
into the new component by combining those of the old components contained in
it. Delete from the topological order all the vertices that are no longer
canonical. Number the remaining canonical vertices consecutively from 1.

{\em Remark}: As in soft-threshold search, using a static strong components
algorithm to find the new component is overkill; a better method is the one described
in the remark before Theorem~\ref{thm:strong-soft-corr}: run a forward
depth-first search from $f(w)$, marking vertices when they are found to be in
the new component.

\begin{theorem}
\label{thm:strong-top-corr}
Maintaining strong components via topological search is correct.

\end{theorem}

\begin{proof}
By induction on the number of arc additions.  Consider the addition of an arc $(v,
w)$.  Let $f$ and $f'$ be the canonical vertex function just before and just after
this addition, respectively.  The addition creates a new component if and only if
$f(v) > f(w)$ and there is a path from $f(w)$ to $f(v)$.  Suppose $f(v) > f(w)$ and
let $F$ and $B$ be the values of the corresponding variables just after the search
stops.  Any path from $f(w)$ to $f(v)$ consists of a sequence of one or more vertices
in $F$ followed by a sequence of one or more vertices in $B$.  Each arc on such a
path is in $Y$.  It follows that the algorithm correctly finds the new component.  If
there is no new component, the algorithm reorders the vertices exactly as in
Section~\ref{sec:top-search} and thus restores topological order.  Suppose there is a
new component.  Let $k$ be the common value of $i$ and $j$ when the search
stops.  The reordering sets $\vertex(k) = f(w)$. This vertex is the canonical vertex of the new component. 
Let $(x, y)$ be an arc. The same argument as in the proof of Theorem~\ref{thm:top-corr} shows
that $f'(x) = f(x) < f(y) = f'(y)$ after the reordering unless $f(x)$ or $f(y)$
or both are in the new component.  If both are in the new component, then $(x,
y)$ is a loop after the addition of $(v, w)$.  Suppose $f(x)$ but not $f(y)$ is
in the new component.  Then $f(x) \in F \cup B$.  If $f(x) \in F$, then
$\position(f(y)) > k$ after the reordering but before the renumbering, so $f'(x)
< f'(y)$ after the reordering and renumbering. If $f(x) \in B$, then $f(y)
\notin B$, since otherwise $f(y)$ is in the component.  It follows that
$\position(f(y)) > k$ before the reordering, and also after the reordering but
before the renumbering, so $f'(x) < f'(y)$ after the reordering and renumbering. 
A symmetric argument applies if $f(y)$ but not $f(x)$ is in the new component. 
\qed

\end{proof}

\begin{theorem}
Maintaining strong components via topological search takes $\bigO(n^{5/2})$ time over
all arc additions.

\end{theorem}

\begin{proof}
The time spent combining rows and columns of $A$ and renumbering vertices after
deletion of non-canonical vertices is $\bigO(n)$ per deleted vertex, totaling
$\bigO(n^2)$ time over all arc additions.  The time spent to find the new component
after a search is $\bigO(|F| + |B|)^2 = \bigO(|F||B|)$ since $|B| \le |F| \le |B| + 1$,
where $F$ and $B$ are the values of the respective variables at the end of the search.
If $x$ is in $F$ and $y$ is in $B$, then either $x$ and $y$ are unrelated before
the arc addition that triggered the search but related after it (and possibly in the
same component), or they are related and in different components before the arc
addition but in the same component after it.  A given pair of
vertices can become related at most once and can be combined into one component at
most once.  There are ${n \choose 2}$ vertex pairs.  Combining these facts, we
find that the total time spent to find new components is $\bigO(n^2)$.

To bound the rest of the computation time, we apply
Theorem~\ref{thm:top-search-time}.  To do this, we modify the strong components
algorithm so that it does not delete non-canonical vertices from the topological
order but leaves them in place.  Such vertices have no incident arcs and are never
moved again.  This only makes the search and reordering time longer, since the
revised algorithm examines non-canonical vertices during search and reordering,
whereas the original algorithm does not.  The proof of
Theorem~\ref{thm:top-search-time} applies to the revised algorithm, giving a bound of
$\bigO(n^{5/2})$ on the time for search and reordering.  \qed

\end{proof}

\section{Remarks} \label{sec:remarks}

We are far from a complete understanding of the incremental topological ordering
problem.  Indeed, we do not even have a tight bound on the running time of
topological search.  Given the connection between this running time and the
$k$-levels problem (see Section~\ref{sec:top-search}), getting a tighter bound seems
a challenging problem.  As mentioned in the introduction, Bender et
al.~\cite{Bender2009} have proposed a completely different algorithm with a running
time of $\Thta(n^2\log n)$.

A more general problem is to find an algorithm that is efficient for any graph
density.  Our lower bound on the number of vertex reorderings is $\Omga(nm^{1/2})$
for any local algorithm (see the end of Section \ref{sec:soft-search}); we conjecture
that there is an algorithm with a matching running time, to within a polylogarithmic
factor.  For sparse graphs, soft-threshold search achieves this bound to within a
constant factor.  For dense graphs, the algorithm of Bender, Fineman, and Gilbert
achieves it to within a logarithmic factor.  For graphs of intermediate density,
nothing interesting is known.

We have used total running time to measure efficiency.  An alternative is to use an
incremental competitive model \cite{Ramalingam1991}, in which the time spent to
handle an arc addition is compared to the minimum work that must be done by any
algorithm, given the same topological order and the same arc addition.  The minimum
work that must be done is the minimum number of vertices that must be reordered,
which is the measure that Ramalingam and Reps used in their lower bound.  (See the
end of Section \ref{sec:soft-search}.)  But no existing algorithm handles an arc
addition in time polynomial in the minimum number of vertices that must be reordered.
To obtain positive results, researchers have compared the performance of their
algorithms to the minimum sum of degrees of reordered vertices \cite{Alpern1990}, or
to a more-refined measure that counts out-degrees of forward vertices and in-degrees
of backward vertices \cite{Pearce2006}.  For these models, appropriately balanced
forms of ordered search are competitive to within a logarithmic factor
\cite{Alpern1990,Pearce2006}.  In such a model, semi-ordered search is competitive to
within a constant factor.  We think, though, that these models are misleading: they
ignore the possibility that different algorithms may maintain different topological
orders, they do not account for the correlated effects of multiple arc additions, and
good bounds have only been obtained for models that overcharge the adversary.

Alpern et al. \cite{Alpern1990} and Pearce and Kelly \cite{Pearce2007} studied
batched arc additions as well as single ones.  Pearce and Kelly give an algorithm
that handles an addition of a batch of arcs in $\bigO(m')$ time, where $m'$ is
the total number of arcs after the addition, and such that the total time for all arc additions
is $\bigO(nm)$.  Thus on each batch the algorithm has the same time bound as a
static algorithm, and the overall time bound is that of the incremental algorithm of
Marchetti-Spaccamela et al. \cite{Marchetti1996}.

This result is not surprising, because {\em any} incremental topological ordering
algorithm can be modified so that each batch of arc additions takes $\bigO(m')$
time but the overall running time increases by at most a constant factor.  The idea is to run
a static algorithm concurrently with the incremental algorithm, each
maintaining its own topological order. Here are the details.  The incremental
algorithm maintains a set of added arcs that have not yet been processed. 
Initially this set is empty.  To handle a new batch of arcs, add them to the
graph and to the set of arcs to be processed.  Then start running a static
algorithm; concurrently, resume the incremental algorithm on the expanded set
of new arcs.  The incremental algorithm deletes an arc at a time from this set
and does the appropriate processing.  Allocate time in equal amounts to the two
algorithms.  If the static algorithm stops before the incremental algorithm
processes all the arcs, suspend the incremental algorithm and use the
topological order computed by the static algorithm as the current order.  If
the incremental algorithm processes all the arcs, stop the static algorithm and
use the topological order computed by the incremental algorithm as the current
order.  This algorithm runs a constant factor slower than the incremental
algorithm and spends $\bigO(m')$ time on each batch of arcs.

For the special case of soft-threshold search, this method can be improved to
maintain a single topological order, and to restart the incremental algorithm
each time the static algorithm completes first.  The time bound remains the
same. If the static algorithm stops first, replace the topological order
maintained by the incremental algorithm by the new one computed by the static
algorithm, and empty the set of new arcs.  These arcs do not need to be
processed by the incremental algorithm.  This works because the running time
analysis of soft-threshold search does not use the current topological order,
only the current graph, specifically the number of related arc pairs.  Whether
something similar works for topological search is open. Much more interesting
would be an overall time bound based on the size and number of batches that is
an improvement for cases other than one batch of $m$ arcs and $m$ batches of
single arcs.

Alpern et al. \cite{Alpern1990} also allowed unrelated vertices to share a position
in the order.  More precisely, their algorithm maintains a numbering of the vertices
such that if $(v, w)$ is an arc, $v$ has a lower number than $w$, but unrelated
vertices may have the same number.  This idea is exploited by Bender, Fineman, and
Gilbert in their new algorithm.

\begin{acks}

{\normalsize 
The last author thanks Deepak Ajwani, whose presentation at the 2007 Data Structures
Workshop in Bertinoro motivated the work described in Section \ref{sec:soft-search},
and Uri Zwick, who pointed out the connection between the running time of topological
search and the $k$-levels problem.  All the authors thank the anonymous referees
for their insightful suggestions that considerably improved the presentation.
}

\end{acks}

%
%

\bibliographystyle{acmsmall}
\bibliography{local}

\end{document}